\documentclass[11pt,letterpaper]{article}

\usepackage[letterpaper,margin=1in]{geometry}
\usepackage[T1]{fontenc}
\usepackage[utf8]{inputenc}
\usepackage{times}
\usepackage{fullpage}
\usepackage{microtype}
\usepackage{setspace}
\usepackage{amsmath,amssymb,amsthm,mathtools}
\usepackage{graphicx}
\usepackage{xcolor}

\usepackage[ruled,vlined,linesnumbered]{algorithm2e}
\SetKw{Return}{return}
\SetKw{Break}{break}
\SetKwProg{WithProb}{\normalfont with probability}{}{end}
\SetKwComment{AlgComment}{$\triangleright$\ }{}

\usepackage{array}
\usepackage{booktabs}
\usepackage{multirow}
\usepackage{tabularx}
\usepackage{threeparttable}
\newcolumntype{Y}{>{\raggedright\arraybackslash}X}

\usepackage{enumitem}
\setlist{topsep=0.5em,itemsep=0.2em,parsep=0pt,partopsep=0pt}

\usepackage[
colorlinks=true,
linkcolor=blue!55!black,
citecolor=blue!55!black,
urlcolor=blue!55!black,
pdfauthor={Qin Zhang},
pdftitle={Optimal Passes and Perfect Sampling for Similarity Graph Statistics}
]{hyperref}

\allowdisplaybreaks
\newtheorem{theorem}{Theorem}[section]
\newtheorem{lemma}[theorem]{Lemma}
\newtheorem{claim}[theorem]{Claim}

\theoremstyle{definition}
\newtheorem{definition}[theorem]{Definition}

\theoremstyle{remark}

\theoremstyle{plain}

\newcommand{\abs}[1]{\left|#1\right|}

\newcommand{\E}{\mathbf{E}}

\renewcommand{\Pr}{\mathbf{Pr}}
\newcommand{\eps}{\epsilon}
\newcommand{\one}{\mathbf{1}}
\newcommand{\1}{\mathbf{1}}
\newcommand{\gray}[1]{\textcolor{gray}{#1}}

\newcommand{\DI}{\mathsf{DI}}
\newcommand{\DDI}{\mathsf{DDI}}
\newcommand{\TSI}{\mathsf{TSI}}
\newcommand{\TV}{d_{\mathrm{TV}}}
\newcommand{\Unif}{\operatorname{Unif}}
\newcommand{\bits}{\{0,1\}}
\newcommand{\FAIL}{\mathtt{null}}
\newcommand{\Null}{\mathtt{null}}

\newcommand{\cD}{\mathcal{D}}
\newcommand{\cE}{\mathcal{E}}
\newcommand{\cH}{\mathcal{H}}
\newcommand{\cL}{\mathcal{L}}
\newcommand{\cU}{\mathcal{U}}
\newcommand{\cB}{\mathcal{B}}

\newcommand{\PiAone}{\Pi^{A}_{1}}
\newcommand{\PiBone}{\Pi^{B}_{1}}
\newcommand{\PiCone}{\Pi^{C}_{1}}
\newcommand{\PiAtwo}{\Pi^{A}_{2}}
\newcommand{\PiBtwo}{\Pi^{B}_{2}}
\newcommand{\Piqry}{\Pi_{\mathrm{qry}}}

\newcommand{\ExactReject}{\mathsf{ExactReject}_p}
\newcommand{\DegreeCoin}{\mathsf{DegreeCoin}}
\newcommand{\High}{\mathsf{H}}
\newcommand{\Low}{\mathsf{L}}
\newcommand{\wt}[1]{\widetilde{#1}}

\makeatletter
\let\oldnl\nl
\newcommand{\nonl}{\renewcommand{\nl}{\let\nl\oldnl}}
\makeatother

\newcommand{\qinsays}[2][]{}

\title{Optimal Passes and Perfect Sampling for Similarity Graph Statistics}

\author{
	Qin Zhang \\
	Computer Science Department\\
	Indiana University\\
	\texttt{qzhangcs@iu.edu}
}

\date{}

\begin{document}
	\hypersetup{pageanchor=false}
	
	\maketitle
	
	\begin{abstract}
		We study statistical estimation on implicit weighted similarity graphs presented as node-arrival streams. Previous work~\cite{LZ26b} obtained constant-pass, sublinear-space algorithms for several basic statistics of such graphs, including the diversity index $\DI=\sum_i d_i^{-1}$ and the degree moments $M_p=\sum_i d_i^p$ for $p>0$, together with their associated sampling problems $L_{\DI}$ and $L_{M_p}$.
		
		We settle three questions left unresolved by the previous work. First, we show that any two-pass streaming algorithm that $1.1$-approximates $\DI$ requires $\Omega(n)$ bits of space; combined with the three-pass $\wt{O}(\sqrt n)$-space algorithm of~\cite{LZ26b}, this establishes a sharp transition between two and three passes for $\DI$. Second and third, for every fixed $c>0$ we give a three-pass $n^{-c}$-perfect $L_{\DI}$-sampler using $O_c(\sqrt n\log n)$ words and a two-pass $n^{-c}$-perfect $L_{M_p}$-sampler using $O_{p,c}(n^{1-1/(p+1)}\log n)$ words. Matching lower bounds up to polylogarithmic factors show that both the number of passes and the polynomial dependence on $n$ are optimal.
	\end{abstract}
	
	
	\pagenumbering{arabic}
	\setcounter{page}{1}
	\hypersetup{pageanchor=true}

\section{Introduction}
\label{sec:intro}

Many fundamental problems in data-stream analysis are built on exact item equality. An item's frequency counts only its identical copies, and classical problems such as distinct elements, frequency moments, and $L_p$-sampling are all defined in terms of these frequencies. For modern data, however, exact equality can be too restrictive. This is particularly relevant for content such as text, images, and videos produced by generative models, where essentially the same content can appear in many different forms. For example, two documents may convey the same information through paraphrasing or summarization, yet exact matching treats them as entirely distinct items. A more natural notion of frequency should allow similar items to contribute to one another's frequencies, with the contribution determined by their degree of similarity.

One way to encode these similarities is through a weighted \emph{similarity graph}: data items are nodes, and the weight between two nodes is their similarity score. The weighted degree of a node is then its similarity-aware frequency. 
Prior work~\cite{LZ26b} initiated a systematic study of statistical problems defined on the similarity graph in the node-arrival data stream model. It gave constant-pass, sublinear-space algorithms for the diversity index ($\DI$), inverse-degree sampling ($L_{\DI}$), degree moments ($M_p$), and degree-moment sampling ($L_{M_p}$), together with lower bounds that are nearly tight in their dependence on the stream length. Several natural questions remained open. First, while diversity-index estimation admits a three-pass sublinear-space algorithm and requires linear space in one pass, its two-pass complexity was left unresolved. Second, the $L_{\DI}$- and $L_{M_p}$-sampling algorithms approximate each target probability only up to a multiplicative factor of $1\pm\eps$, leaving open whether perfect sampling is possible with the same number of passes and comparable space. This paper resolves these questions: we establish a sharp separation between two and three passes for diversity-index estimation, and give perfect $L_{\DI}$- and $L_{M_p}$-samplers with inverse-polynomially small failure probability and total-variation error.

\vspace{2mm}
\noindent{\bf Model and problems.\ }
Let $sim:\cU\times\cU\to[0,1]$ be a symmetric similarity function satisfying $sim(x,x)=1$ for every $x\in\cU$.  
The input is a node-arrival stream $\sigma=(\sigma_1,\ldots,\sigma_n)\in\cU^n$.  Each stream item is regarded as a distinct node, even when two items have the same label.  Write $[n]=\{1,\ldots,n\}$. The stream induces a complete weighted graph $G_\sigma$ on node set $[n]$, in which the pair $(i,j)$ has weight
$
w(i,j)=sim(\sigma_i,\sigma_j).
$
In particular, every node has a self-loop of unit weight.  The weighted degree of node $\sigma_j$ is
$$
d_j=d(\sigma_j) \coloneqq \sum_{i \in [n]}sim(\sigma_i,\sigma_j).
$$
Clearly, $1\le d_j\le n$.

A streaming algorithm makes one or more sequential passes over the node stream.  When $\sigma_j$ arrives, the algorithm may evaluate its similarity to nodes currently stored in memory through the known similarity function $sim(\cdot, \cdot)$.  Our goal is to solve the following problems while minimizing the space usage and the number of passes.
\begin{itemize}
	\item \emph{Diversity index.}  $\DI(\sigma)=\sum_{j \in [n]}d_j^{-1}.$
	\smallskip
	
	\item \emph{$L_{\DI}$-sampling.} The {\em $L_{\DI}$ distribution} on stream nodes is
	$
	\pi_{\DI}(j)=\frac{1}{\DI(\sigma)} d_j^{-1}
	$
	for any $j \in [n]$.
	The $L_{\DI}$-sampling problem asks the algorithm to output a node $j \in [n]$ according to this distribution.
	\smallskip
	
	\item \emph{$L_{M_p}$-sampling.} For a fixed constant $p>0$, define the $p$-th degree moment
	$
	M_p(\sigma)=\sum_{j \in [n]} d_j^p.
	$
	The {\em $L_{M_p}$ distribution} on stream nodes is
	$
	\pi_{M_p}(j)=\frac{1}{M_p(\sigma)}d_j^p.
	$
	The $L_{M_p}$-sampling problem asks the algorithm to output a node $j \in [n]$  according to this distribution.
\end{itemize}
When the stream $\sigma$ is clear, we write $\DI(\sigma)$ and $M_p(\sigma)$ as $\DI$ and $M_p$.

These definitions simultaneously generalize several classical stream statistics.  If $sim(\cdot, \cdot)$ is the equality function and an item $x$ occurs $f_x$ times, then each occurrence of $x$ has degree $f_x$.
Consequently, $\DI=\sum_x f_x f_x^{-1} = \sum_x 1$ is the number of distinct items, while $M_p=\sum_x f_x f_x^p = \sum_x f_x^{p+1}$ is the classical frequency moment $F_{p+1}$. Similarly, $L_{\DI}$-sampling becomes uniform sampling from the distinct item values, and $L_{M_p}$-sampling becomes
classical $L_{p+1}$-sampling. 
\smallskip

For a random variable $X$, let $\mathcal L(X)$ denote its distribution. For two finite distributions, their total variation distance is defined as $\TV(P,Q)=\frac12\sum_x\abs{P(x)-Q(x)}$. For the sampling problems, we use the following guarantees.
\begin{definition}[Perfect sampler]
	\label{def:perfect-sampling}
	Let $\pi$ be a target distribution on the stream nodes.  A
	streaming algorithm is a $\delta$-perfect sampler for $\pi$ if it outputs a
	stream node or $\FAIL$, and
	$$
	\Pr[J=\FAIL]\le\delta \qquad\text{and}\qquad \TV\bigl(\mathcal L(J\mid J\ne\FAIL),\pi\bigr)\le\delta.
	$$
	When this holds with $\delta=O(n^{-c})$ for any desired fixed constant $c>0$, we omit $\delta$ and simply call the algorithm a perfect sampler.
\end{definition}

\vspace{1mm}
\noindent{\bf Our results.\ }
We call $\widetilde X$ a $(1+\eps)$-approximation to $X\ge0$ if $\widetilde X=(1\pm\eps)X$, and write $x=a\pm b$ for $x\in[a-b,a+b]$. Subscripts in $O_c(\cdot)$ and $\Omega_c(\cdot)$ indicate constant dependence on $c$. Throughout, $\log$ denotes the logarithm to base $2$ and $\ln$ the natural logarithm. Upper bounds are measured in words and lower bounds in bits, with each node,
similarity value, counter, and $O(\log n)$-bit random value occupying $O(1)$ words. Our contributions are as follows.

For upper bounds, we work in the real-RAM model with free random bits; both assumptions are removed by standard techniques (discretization to $O(\log n)$ bits, and Nisan's generator~\cite{Nisan90} for the one random map that is too large to store, see Section~\ref{sec:L-DI-sampling}), with
inverse-polynomially small additional error absorbed into the $n^{-c}$ guarantee of Definition~\ref{def:perfect-sampling}.

\smallskip
\noindent{\bf Result 1: A sharp two-to-three-pass transition for diversity-index estimation.\ } We prove that every two-pass streaming algorithm that returns a
$1.1$-approximation to $\DI$ with probability at least $0.51$ requires $\Omega(n)$ bits of space, even for Boolean similarity and $O(\log n)$-bit node representations (Theorem~\ref{thm:streaming-main}).

Combined with the three-pass $\widetilde O(\sqrt n)$-space upper bound and the constant-pass $\Omega(\sqrt n)$-space lower bound of~\cite{LZ26b}, this nearly completes the pass--space landscape for constant accuracy: one and two passes require linear space; at three passes the complexity drops to $\widetilde\Theta(\sqrt n)$; and additional constant passes cannot improve the polynomial dependence on $n$.  This is an unusually abrupt pass transition; we are not aware of another natural statistic whose complexity drops from linear to $\sqrt{n}$ between two and three passes.

\smallskip
\noindent{\bf Result 2: Three-pass perfect $L_{\DI}$-sampling.\ }  For every fixed $c>0$, we give a three-pass $n^{-c}$-perfect $L_{\DI}$-sampler using $O_c(\sqrt n\log n)$ words of space (Theorem~\ref{thm:L-DI-upper}).  The sampler
of~\cite{LZ26b} had the same polynomial dependence on $n$ but incurred a multiplicative $1\pm\eps$ distortion in its output probabilities.

We also prove that every two-pass perfect $L_{\DI}$-sampler with constant failure probability requires $\Omega(n/\log^2 n)$ bits (Theorem~\ref{thm:L-DI-lb}), and that any constant-pass perfect sampler requires $\Omega(\sqrt n)$ bits (Theorem~\ref{thm:L-DI-constant-pass-lb}). Hence, the upper bound is optimal in its polynomial dependence on $n$, and three passes are necessary for sublinear-space perfect $L_{\DI}$-sampling, up to logarithmic factors.

\smallskip
\noindent{\bf Result 3: Two-pass perfect $L_{M_p}$-sampling.\ } For every fixed $p>0$, and every fixed $c>0$, we give a two-pass $n^{-c}$-perfect $L_{M_p}$-sampler using
$O_{p,c}\left(n^{1-1/(p+1)}\log n\right)$ words of space
(Theorem~\ref{thm:main}).  This improves the
two-pass $(1\pm\eps)$-multiplicatively approximate sampler of~\cite{LZ26b} to a perfect sampler, without changing the polynomial dependence on $n$.

Conversely, one pass requires $\Omega_p(n/\log^2 n)$ bits (Theorem~\ref{thm:LMp-one-pass-lb}), and every constant number of passes requires
$\Omega_p\left({n^{1-1/(p+1)}}/{\log^2 n}\right)$
bits (Theorem~\ref{thm:LMp-constant-pass-lb}).  Hence, the two-pass algorithm is optimal in the exponent of $n$, up to polylogarithmic factors.

\vspace{2mm}
\noindent{\bf Related work.}
We review several lines of work related to ours.

\smallskip
\noindent{\em The similarity-graph framework.\ } This paper builds on the similarity-graph framework introduced in \cite{LZ26b}.  The work \cite{LZ26b} defined the diversity index, degree moments, and their associated sampling problems.  This paper settles three questions left open there. Two recent works~\cite{Zhang26a,Zhang26b} studied two specific similarity functions, cosine similarity and the Gaussian kernel, in the same framework. They show that when the similarity function possesses suitable algebraic or geometric structure, one-pass sublinear-space algorithms are possible for estimating the diversity index and degree moments.

\smallskip
\noindent{\em Classical stream statistics.\ } Distinct elements and frequency moments are foundational problems in streaming
algorithms~\cite{AMS99,BJKST02,KNW10,IW05}.  Sampling from a frequency vector has likewise been studied extensively, including $\ell_0$-sampling and inverse sampling~\cite{Gibbons01,FIS05,CMR05}, relative-error $\ell_p$-sampling~\cite{MW10,AKO11,JST11},
and perfect $\ell_p$-sampling~\cite{JW18,JWZ22,WXZ25,SWZ25}.  When $sim(\cdot, \cdot)$ is the equality function, $\DI$ becomes $F_0$, $M_p$ becomes $F_{p+1}$, and our two node-sampling distributions induce $\ell_0$- and $\ell_{p+1}$-sampling on item values.  The access models are nevertheless quite different: in a classical
stream, each update identifies a coordinate of the frequency vector, whereas a node's weighted degree here is an implicit sum of pairwise similarities that
can be evaluated only against nodes stored in memory.  This restricted access leads to polynomial space requirements in our setting.

\smallskip
\noindent{\em Similarity-aware analysis of noisy streams.\ }
A related line of work studies classical stream statistics when nearby but nonidentical items should be treated as copies of the same entity.  Chen and
Zhang~\cite{CZ16} studied robust distinct-elements estimation in constant-dimensional Euclidean spaces and in metric spaces supporting suitable locality-sensitive hashing.  The follow-up work~\cite{CZ18} considered the
corresponding distinct-sampling problem, and Zhang~\cite{Zhang25} extended both
distinct-elements estimation and $\ell_0$-sampling to general metric spaces. Liu and Zhang~\cite{LZ26Mismatch} proposed a different noisy-data framework for estimating frequency moments of an unobserved ground-truth dataset, with the
approximation controlled by a mismatch-ambiguity parameter.  In contrast, this paper takes the observed weighted similarity graph itself as the object of computation.

\smallskip
\noindent{\em Node-arrival and graph-access models.\ }
Most graph-streaming research assumes an edge-arrival stream; see McGregor~\cite{McGregor14} for a survey.  In the node-arrival model, Emek, Halld\'orsson, and Ros\'en~\cite{EHR16}, and Cabello and P\'erez-Lantero~\cite{CP17} studied interval selection, where the stream consists of intervals whose intersections implicitly define the graph.  Cormode, Dark, and Konrad~\cite{CDK18} studied the
Caro--Wei bound in edge-arrival and {\em explicit} node-arrival streams, where an arriving node is accompanied by its edges to earlier nodes.  More
recently, Liu, Villalobos and Zhang~\cite{LVZ26} considered correlation clustering cost in node-arrival streams.  

\smallskip
\noindent{\em Degree statistics.\ }
Average degree, star counts, and higher degree-distribution moments have also been studied in sublinear-time graph-query models~\cite{Feige04,GR08,GRS11,ERS17}.
Those models provide combinations of uniform-node, degree, and neighbor queries and therefore differ substantially from our pairwise-access streaming
model.

\vspace{2mm}
\noindent{\bf Roadmap.\ }
Section~\ref{sec:DI-lower-bound} proves the two-pass lower bound for $\DI$. Section~\ref{sec:L-DI-sampling} gives the $L_{\DI}$ sampler and its lower bounds.  Section~\ref{sec:LMp-sampling} gives the two-pass $L_{M_p}$ sampler and the matching lower bounds.  Some probability and information-theoretic tools can be found in Appendix~\ref{sec:preliminaries}.

\section{A Two-Pass Lower Bound for Diversity-Index Estimation}
\label{sec:DI-lower-bound}

In this section, we show that any sublinear-space algorithm for accurately estimating $\DI$ requires at least three passes.

\begin{theorem}
	\label{thm:streaming-main}
	Any randomized two-pass streaming algorithm that, on every node-arrival stream $\sigma$ of length $n$, returns a $1.1$-approximation to $\DI(\sigma)$ with probability at least $0.51$ requires $\Omega(n)$ bits of space.  The lower bound holds even when the similarity function is Boolean and every stream item has an $O(\log n)$-bit representation.
\end{theorem}

For the remainder of this section, set $m=n/41$. We omit floors and ceilings whenever they do not affect the asymptotic bounds.

\subsection{Proof Overview}

We prove Theorem~\ref{thm:streaming-main} through a three-party communication problem called \emph{Two-Step Indexing} ($\TSI$).  Alice receives a random bit vector $X=(X_1,\ldots,X_m)\in\bits^m$, Bob receives a random {\em pointer} vector $Y=(Y_1,\ldots,Y_m)\in[m]^m$, and Charlie receives a random index $Z\in[m]$.  Their goal is to output $X_{Y_Z}.$ The communication order is
$
\text{Alice}\to\text{Bob}\to\text{Charlie}
\to\text{Alice}\to\text{Bob}\to\text{Charlie}\to\text{output},
$
which is precisely the order induced by two passes over a stream partitioned into three contiguous blocks.

Viewing the stream as a graph, we partition the hard instance into three node sets: Alice creates one node for each index $i\in[m]$; Bob creates $\gamma:=20$ nodes for each index $q\in[m]$; and Charlie creates a clique on $\gamma m$ nodes. For every $q\neq Z$, each of Bob's $\gamma$ nodes corresponding to $q$ is adjacent to every node in Charlie's clique, and therefore has large degree. In contrast, the $\gamma$ nodes corresponding to $q=Z$ have no neighbors among Charlie's nodes. Each such node has degree $1+X_{Y_Z}.$ Hence, the $\gamma$ Bob nodes corresponding to $q=Z$ contribute $\frac{\gamma}{1+X_{Y_Z}}$ to the diversity index; their total contribution is $\gamma$ when $X_{Y_Z}=0$ and $\gamma/2$ when $X_{Y_Z}=1$. The total contribution from all remaining nodes is bounded by $3$. Consequently, if $X_{Y_Z}=0$, then $\DI\ge20$, whereas if $X_{Y_Z}=1$, then $\DI\le13$. Therefore, any $1.1$-approximation to $\DI$ determines the value of $X_{Y_Z}$.

We now return to the $\TSI$ problem. The communication bottleneck can be seen as follows. Bob sends his first message before learning Charlie's index $Z$, and hence before knowing which entry $J=Y_Z$ will eventually be relevant. Consequently, a short first message cannot reveal much information about this random pointer. After the first round, Alice may learn $Z$, but she still does not know the corresponding value $J=Y_Z$. Thus, across her two messages, Alice can convey only limited information about the randomly chosen query bit $X_J$. Therefore, no two-round protocol with short messages exists for $\TSI$. 

\subsection{Communication Problems and Hard Distributions}

We now formalize the above intuition, beginning with the problem definition and the construction of the hard instance.

\begin{definition}[Two-Step Indexing ($\TSI_m$)]
	\label{def:tsi}
	Alice, Bob, and Charlie receive $X\in\bits^m$,  $Y\in[m]^m$, and $Z\in[m]$, respectively, and their goal is to output $X_{Y_Z}$. They communicate in two rounds in the
	order
	$$
	\text{Alice}\xrightarrow{\PiAone}
	\text{Bob}\xrightarrow{\PiBone}
	\text{Charlie}\xrightarrow{\PiCone}
	\text{Alice}\xrightarrow{\PiAtwo}
	\text{Bob}\xrightarrow{\PiBtwo}
	\text{Charlie}\to\text{output}.
	$$
	We write
	$
	\Pi=\PiAone\circ\PiBone\circ\PiCone\circ\PiAtwo\circ\PiBtwo
	$
	for the communication transcript, excluding Charlie's final output.
\end{definition}

Two-Step Indexing is reminiscent of pointer chasing~\cite{PRV01}, as both problems involve a sequence of dependent lookups. However, the standard pointer-chasing lower bounds do not directly apply to $\TSI$: the latter has a different three-party input partition and communication order, and its final lookup returns a bit held by the first party. We therefore prove the communication lower bound for $\TSI$ directly using an information-theoretic argument.

We consider the following hard distribution for $\TSI_m$.

\begin{definition}[Distribution $\cD_{\mathrm{TSI}}$]
	In this distribution, $X$ is sampled uniformly at random from $\bits^m$,	$Y_1,\ldots,Y_m$ are sampled independently, uniformly at random from $[m]$, and
	$Z$ is sampled uniformly at random from $[m]$. 
\end{definition}

Let $J\coloneqq Y_Z$. Under $\cD_{\mathrm{TSI}}$, $J$ is uniform on $[m]$ and is independent of both $X$ and $Z$.  We will prove in the next subsection the following communication complexity bound for $\TSI_m$.

\begin{theorem}
	\label{thm:tsi-lb}
	For every constant $\delta<1/2$, any randomized two-round protocol for $\TSI_m$ with error at most $\delta$ has communication cost $\Omega(m)$ bits.
\end{theorem}

We next define the graph communication problem used in the reduction.

\begin{definition}[Distributed Diversity Index ($\DDI_m$)]
	\label{def:ddi}
	Alice, Bob, and Charlie hold node sets $V_A,V_B,V_C$, respectively.  A fixed symmetric similarity function $f$ is known to all three parties.  They communicate in the same
	two-round order as in Definition~\ref{def:tsi}, and Charlie must output a multiplicative approximation to the diversity index of the graph induced by $V_A\cup V_B\cup V_C$.
\end{definition}

Fix $\gamma=20$.  Consider the universe consisting of the following labeled items:
\begin{eqnarray*}
	\cU_A&=&\{(\mathsf A,i,x):i\in[m],\ x\in\bits\},\\
	\cU_B&=&\{(\mathsf B,j,q,r):j,q\in[m],\ r\in[\gamma]\},\\
	\cU_C&=&\{(\mathsf C,z,\ell):z\in[m],\ \ell\in[\gamma m]\}.
\end{eqnarray*}

We define a Boolean symmetric similarity function $f$ as follows.  Every
item is similar to itself.  For distinct items:
\begin{enumerate}
	\item all $\mathsf A$-type items are mutually similar;
	\item all $\mathsf C$-type items are mutually similar;
	\item
	$f\bigl((\mathsf A,i,x),(\mathsf B,j,q,r)\bigr)=1$ if and only if
	$i=j\ \text{and}\ x=1$;
	\item
	$f\bigl((\mathsf B,j,q,r),(\mathsf C,z,\ell)\bigr)=1$ if and only if 
	$q\neq z$;
	\item every other pair of distinct items has similarity zero.
\end{enumerate}

\begin{definition}[Distribution $\cD_{\mathrm{DDI}}$]
	\label{def:Dddi}
	Given $(X,Y,Z)\sim\cD_{\mathrm{TSI}}$, assign Alice, Bob, and Charlie the following node sets respectively: 
	\begin{eqnarray*}
		V_A(X)&=&\{a_i=(\mathsf A,i,X_i):i\in[m]\},\\
		V_B(Y)&=&\{b_{q,r}=(\mathsf B,Y_q,q,r):q\in[m],\ r\in[\gamma]\},\\
		V_C(Z)&=&\{c_\ell=(\mathsf C,Z,\ell):\ell\in[\gamma m]\}.
	\end{eqnarray*}
	Let $\cD_{\mathrm{DDI}}$ denote the induced distribution on
	$(V_A,V_B,V_C)$. 
\end{definition}

The total number of nodes is
$
(2\gamma+1)m=41m = n.
$
Each node label contains only a constant number of integers in $\{1, \ldots, \gamma m\}$ where $\gamma m = \Theta(n)$,
so it has an $O(\log n)$-bit representation; moreover, $f$ can be evaluated
directly from two labels.

We will show in Section~\ref{sec:reduction-TSI-DDI} that Theorem~\ref{thm:tsi-lb} implies the communication lower bound for  $\DDI_m$.
\begin{theorem}
	\label{thm:ddi-lb}
	Any randomized two-round protocol that, on the distribution
	$\cD_{\mathrm{DDI}}$, returns a $1.1$-approximation to the diversity index
	with error probability at most $0.49$ requires $\Omega(m)$ bits of
	communication.
\end{theorem}

This result immediately gives the corresponding space lower bound for streaming diversity index.

\begin{proof}[Proof of Theorem~\ref{thm:streaming-main}]
	Suppose a two-pass streaming algorithm $\mathcal A$ uses $s$ bits of space and returns a $1.1$-approximation with error at most $0.49$.
	Alice first runs $\mathcal A$ on the block $V_A$ and sends the resulting $s$-bit memory state $\PiAone$ to Bob.  Bob continues on $V_B$ and sends the state $\PiBone$ to Charlie.  Charlie continues on $V_C$ and sends the state $\PiCone$ back to Alice.  In the second pass, Alice again processes $V_A$ and sends $\PiAtwo$ to Bob; Bob processes $V_B$ and sends
	$\PiBtwo$ to Charlie; Charlie processes $V_C$ and outputs whatever $\mathcal A$ outputs.  The resulting two-round $\DDI$ protocol communicates at most
	$5s$ bits.  Theorem~\ref{thm:ddi-lb} gives $5s=\Omega(m)$, so $s=\Omega(m)=\Omega(n)$.  
	
	Finally, by the definition of the distribution $\cD_{\mathrm{DDI}}$, the similarity function is Boolean, and each node can be represented using $O(\log n)$ bits.
\end{proof}

\subsection{Proof of Theorem~\ref{thm:tsi-lb}}

By Yao's minimax principle, it is enough to prove that every deterministic protocol having error at most $\delta$ under $\cD_{\mathrm{TSI}}$ uses $\Omega(m)$ bits.  Fix such a deterministic protocol.  

For a random message $M$, we use $|M|$ to denote the worst-case length (in bits) of $M$ over all inputs in the support of the hard distribution. We assume without loss of generality that every message is padded to its worst-case length. It is enough to prove that
$|\PiAone|+|\PiBone|+|\PiAtwo|=\Omega(m).$

\vspace{2mm}
\noindent{\bf Protocol augmentation.\ }
Define the \emph{query transcript}
$$
\Piqry\coloneqq(\PiAone,\PiBone,Z).
$$
We augment the second round of the protocol as follows: First, before Alice sends $\PiAtwo$, we reveal the entire query transcript $\Piqry$ to her.  This can only help Alice, since in the original protocol she only receives $\PiCone$, which is a deterministic function of $(\PiAone,\PiBone,Z)$, and thus can be
reconstructed from $\Piqry$. Second, we reveal $(Y,\Piqry,\PiAtwo)$ to Charlie at the end of the second round before he outputs the answer.  This can only
help Charlie. Therefore there is a decoder
$
\widehat W=\widehat W(Y,\Piqry,\PiAtwo)
$
that predicts
$
W\coloneqq X_J=X_{Y_Z}
$
with error probability at most $\delta$.
\smallskip

The following lemma shows that if the message sizes are small, then after the first round, the query index $J$ is still almost uniform.
\begin{lemma}
	\label{lem:query-nearly-uniform}
	Under the setup above,
	$
	I(J;\Piqry)\leq {|\PiBone|}/{m}.
	$
	Consequently, letting
	$
	p_\pi(j)\coloneqq\Pr[J=j\mid\Piqry=\pi]
	$
	and $U_m$ denote the uniform distribution on $[m]$,
	\begin{equation*}
		\label{eq:query-nearly-uniform}
		\E_{\Piqry}\bigl[\TV(p_{\Piqry},U_m)\bigr] \leq
		\sqrt{\frac{\ln2}{2}\,\frac{|\PiBone|}{m}}.
	\end{equation*}
\end{lemma}

\begin{proof}
	Because $\PiAone$ is a deterministic function of $X$, and $(Y,Z)$ is
	independent of $X$, the message $\PiAone$ is independent of $(J,Z)$.
	Moreover, $J=Y_Z$ is independent of $Z$, since for every $j,z\in[m]$,
	$$
	\Pr[J=j\mid Z=z] = \Pr[Y_z=j] = \frac1m = \Pr[J=j].
	$$
	Thus $I(J;\PiAone,Z)=0$, and the chain rule gives
	\begin{equation}
		\label{eq:why-chain-rule}
		I(J;\Piqry) = I(J;\PiAone,\PiBone,Z) = I(J;\PiAone,Z) + I(J;\PiBone\mid\PiAone,Z) = I(J;\PiBone\mid\PiAone,Z).
	\end{equation}
	
	Since $Z$ is uniform and independent of $(Y,\PiAone,\PiBone)$,
	conditioning on $Z=i$ simply selects the coordinate $Y_i$.  Therefore,
	\begin{equation}
		\label{eq:random-coordinate-average}
		I(J;\PiBone\mid\PiAone,Z) = I(Y_Z;\PiBone\mid\PiAone,Z)
		= \frac1m\sum_{i\in[m]} I(Y_i;\PiBone\mid\PiAone).
	\end{equation}
	
	The coordinates of $Y$ remain mutually independent after conditioning on $\PiAone$, because $\PiAone$ depends only on the independent variable
	$X$.  Therefore, 
	\begin{eqnarray}
		\sum_{i\in[m]}I(Y_i;\PiBone\mid\PiAone)
		&=&H(Y\mid\PiAone)
		-\sum_{i\in[m]}H(Y_i\mid\PiAone,\PiBone) \nonumber \\
		&\leq& H(Y\mid\PiAone)-H(Y\mid\PiAone,\PiBone) \nonumber \\
		&=&I(Y;\PiBone\mid\PiAone) \leq H(\PiBone)\leq |\PiBone|,
		\label{eq:independent-coordinate-info}
	\end{eqnarray}
	where the first inequality uses subadditivity of conditional entropy.  Combining
	\eqref{eq:why-chain-rule}, \eqref{eq:random-coordinate-average}, and \eqref{eq:independent-coordinate-info} proves
	$I(J;\Piqry)\leq {|\PiBone|}/{m}$.
	
	Since $J$ is uniformly distributed, we have
	$
	I(J;\Piqry) =\E_{\Piqry} \left[D_{\mathrm{KL}} \bigl(p_{\Piqry}\|U_m\bigr)\right].
	$
	Applying Pinsker's inequality pointwise and then Jensen's inequality gives
	\begin{eqnarray*}
		\E_{\Piqry}\bigl[\TV(p_{\Piqry},U_m)\bigr]
		&\leq& \E_{\Piqry}
		\left[\sqrt{\frac{\ln2}{2} D_{\mathrm{KL}}(p_{\Piqry}\Vert U_m)}\right] \\
		&\leq& \sqrt{\frac{\ln2}{2}
			\E_{\Piqry}\bigl[D_{\mathrm{KL}}(p_{\Piqry}\Vert U_m)\bigr]} \\
		&=& \sqrt{\frac{\ln2}{2} I(J;\Piqry)}
		\leq \sqrt{\frac{\ln2}{2}\cdot\frac{|\PiBone|}{m}},
	\end{eqnarray*}
	which proves the second item of the lemma.
\end{proof}

We use the following independence properties in the communication process.
\begin{lemma}
	\label{lem:query-conditional-independence}
	For the augmented protocol, we have
	$
	(X,\PiAtwo)\perp Y\mid\Piqry.
	$
\end{lemma}

\begin{proof}
	Fix a query transcript
	$\pi=(\pi^A_1,\pi^B_1,z)$ of positive probability.  Since the protocol is
	deterministic, there are functions $g_A$ and $g_B$ such that
	$\PiAone=g_A(X)$, and
	$  \PiBone=g_B(Y,\PiAone)$.
	For notational convenience, define
	$$
	\mu_\pi(x) \coloneqq\Pr[X=x]\,\1\{g_A(x)=\pi^A_1\}
	$$
	and
	$$
	\nu_\pi(y)
	\coloneqq\Pr[Y=y]\,\1\{g_B(y,\pi^A_1)=\pi^B_1\}.
	$$
	Since $X,Y,Z$ are independent, we have
	\begin{equation}
		\label{eq:perp-1}
		\Pr[X=x,Y=y,\Piqry=\pi] = \Pr[Z=z]\mu_\pi(x)\nu_\pi(y).
	\end{equation}
	On the other hand,
	\begin{equation}
		\label{eq:perp-2}
		\Pr[\Piqry=\pi] = \Pr[Z=z] \left(\sum_{x'}\mu_\pi(x')\right) \left(\sum_{y'}\nu_\pi(y')\right).
	\end{equation}
	Dividing \eqref{eq:perp-1} by \eqref{eq:perp-2}, we obtain
	\begin{eqnarray*}
		\Pr[X=x,Y=y\mid\Piqry=\pi]
		&=&
		\frac{\mu_\pi(x)}{\sum_{x'}\mu_\pi(x')} \cdot \frac{\nu_\pi(y)}{\sum_{y'}\nu_\pi(y')}\\
		&=&\Pr[X=x\mid\Piqry=\pi] \Pr[Y=y\mid\Piqry=\pi].
	\end{eqnarray*}
	
	In the augmented protocol, $\PiAtwo$ is a deterministic function of $(X,\Piqry)$.  Once $\Piqry$ is fixed, adding such a function to the $X$ side cannot create dependence with $Y$.  This gives the lemma.
\end{proof}

The following lemma shows that if the messages are of small sizes, then Alice reveals limited information about the queried bit.
\begin{lemma}
	\label{lem:alice-target-information}
	For the augmented protocol, we have
	\begin{equation*}
		\label{eq:final-info-upper}
		1-H(X_J\mid Y,\Piqry,\PiAtwo) \leq \frac{|\PiAone|+|\PiAtwo|}{m} + \sqrt{\frac{\ln2}{2}\,\frac{|\PiBone|}{m}}.
	\end{equation*}
\end{lemma}

\begin{proof}
	For every transcript value $\pi$ and coordinate $j\in[m]$, define
	$$
	\lambda_j(\pi) \coloneqq 1-H(X_j\mid\Piqry=\pi,\PiAtwo).
	$$
	Since $X_j$ is a bit,
	$0\leq \lambda_j(\pi)\leq1$.
	
	We first show the following identity
	\begin{equation}
		\label{eq:target-entropy-expansion}
		1-H(X_J\mid Y,\Piqry,\PiAtwo) =\E_{\Piqry} \Bigg[\sum_{j\in[m]}p_{\Piqry}(j)\lambda_j(\Piqry)\Bigg].
	\end{equation}
	
	We start with some preparation. By Lemma~\ref{lem:query-conditional-independence}, marginalizing over $X$  gives
	$
	\PiAtwo\perp Y\mid\Piqry.
	$
	Since $Z$ is contained in $\Piqry$ and $J=Y_Z$, it follows that
	\begin{equation}
		\label{eq:J-given-query-transcript}
		\Pr[J=j\mid\Piqry=\pi,\PiAtwo=\alpha] = \Pr[J=j\mid\Piqry=\pi]
		=p_\pi(j).
	\end{equation}
	Again by Lemma~\ref{lem:query-conditional-independence}, conditioning on $\PiAtwo$ gives
	$X\perp Y\mid(\Piqry,\PiAtwo)$. Hence, for every $j\in[m]$,
	\begin{equation}
		\label{eq:Xj-independent-of-Y}
		H(X_j\mid Y=y,\Piqry=\pi,\PiAtwo=\alpha)
		=H(X_j\mid\Piqry=\pi,\PiAtwo=\alpha).
	\end{equation}
	
	We now investigate the left-hand side of \eqref{eq:target-entropy-expansion}.
	Conditioned on $(Y=y,\Piqry=\pi)$, the index $J=Y_Z$ is determined, since $Z$ is part of $\Piqry$.  Expanding the conditional entropy over $(Y,\Piqry,\PiAtwo)$ therefore gives
	\begin{equation*}
		H(X_J\mid Y,\Piqry,\PiAtwo) = \sum_{y,\pi,\alpha}
		\Pr[Y=y,\Piqry=\pi,\PiAtwo=\alpha] H(X_{y_z}\mid Y=y,\Piqry=\pi,\PiAtwo=\alpha),
	\end{equation*}
	where $z$ denotes the value of $Z$ contained in $\pi$.
	By \eqref{eq:Xj-independent-of-Y}, conditioning on $Y=y$ does not affect the
	conditional distribution of $X_j$ once $(\Piqry,\PiAtwo)$ is fixed.  Hence,
	$$
	H(X_{y_z}\mid Y=y,\Piqry=\pi,\PiAtwo=\alpha) = H(X_{y_z}\mid\Piqry=\pi,\PiAtwo=\alpha).
	$$
	Grouping the values of $y$ according to $j=y_z$ now yields
	\begin{eqnarray*}
		&&H(X_J\mid Y,\Piqry,\PiAtwo)\\
		&=&
		\sum_{\pi,\alpha} \Pr[\Piqry=\pi,\PiAtwo=\alpha] \sum_{j\in[m]}	\Pr[J=j\mid\Piqry=\pi,\PiAtwo=\alpha]
		H(X_j\mid\Piqry=\pi,\PiAtwo=\alpha)\\
		&\stackrel{\text{by}~\eqref{eq:J-given-query-transcript}}{=}&
		\sum_{\pi,\alpha} \Pr[\Piqry=\pi,\PiAtwo=\alpha] \sum_{j\in[m]}p_\pi(j) H(X_j\mid\Piqry=\pi,\PiAtwo=\alpha)\\
		&=&
		\sum_\pi \Pr[\Piqry=\pi] \sum_{j\in[m]}p_\pi(j) \sum_\alpha	\Pr[\PiAtwo=\alpha\mid\Piqry=\pi]\,
		H(X_j\mid\Piqry=\pi,\PiAtwo=\alpha)\\
		&=&
		\E_{\Piqry}\Bigg[\sum_{j\in[m]}p_{\Piqry}(j) H(X_j\mid\Piqry,\PiAtwo)\Bigg] = \E_{\Piqry}\Bigg[\sum_{j\in[m]}p_{\Piqry}(j) \left(1 -  \lambda_j(\Piqry) \right)\Bigg].
	\end{eqnarray*}
	Since $\sum_{j\in[m]}p_\pi(j)=1$, this is equivalent to
	\eqref{eq:target-entropy-expansion}.
	
	For every fixed $\pi$, the definition of total variation distance gives
	\begin{equation}
		\sum_{j\in[m]}p_\pi(j)\lambda_j(\pi) \leq \frac1m\sum_{j\in[m]}\lambda_j(\pi) + \TV(p_\pi,U_m).
		\label{eq:tv-test-function}
	\end{equation}
	Combining \eqref{eq:target-entropy-expansion} and
	\eqref{eq:tv-test-function} (after taking expectation over $\Piqry$) gives
	\begin{equation}
		\label{eq:alice-info-reveal}
		1-H(X_J\mid Y,\Piqry,\PiAtwo) \leq
		\frac1m \E_{\Piqry} \Bigg[\sum_{j\in[m]}\lambda_j(\Piqry)\Bigg] + \E_{\Piqry} \left[\TV(p_{\Piqry},U_m)\right].
	\end{equation}
	
	We bound the right-hand side of \eqref{eq:alice-info-reveal} in two parts.  For the second part, Lemma~\ref{lem:query-nearly-uniform} gives
	\begin{equation}
		\label{eq:TV}
		\E_{\Piqry}\left[\TV(p_{\Piqry},U_m)\right] \le \sqrt{\frac{\ln2}{2}\,\frac{|\PiBone|}{m}}.
	\end{equation}
	
	For the first part, by subadditivity of conditional entropy,
	\begin{eqnarray}
		\E_{\Piqry}
		\Bigg[\sum_{j\in[m]}\lambda_j(\Piqry)\Bigg]
		&=& m-\sum_{j\in[m]}H(X_j\mid\Piqry,\PiAtwo)\nonumber\\
		&\leq& m-H(X\mid\Piqry,\PiAtwo)\nonumber\\
		&=& I(X;\Piqry,\PiAtwo),
		\label{eq:uniform-coordinate-info}
	\end{eqnarray}
	where the last equality uses $H(X)=m$.
	
	To bound this mutual information, recall that
	$\Piqry=(\PiAone,\PiBone,Z)$.  The chain rule gives
	\begin{equation*}
		I(X;\Piqry,\PiAtwo) = I(X;\PiAone) + I(X;\PiBone,Z\mid\PiAone)+I(X;\PiAtwo\mid\Piqry).
		\label{eq:alice-info-chain}
	\end{equation*}
	The middle term is zero.  Indeed, conditioning on $\PiAone$ preserves the
	independence of $X$ and $(Y,Z)$, because $\PiAone$ is a function only of
	$X$.  Once $\PiAone$ is fixed, $(\PiBone,Z)$ is a function only of
	$(Y,Z)$, so
	$
	I(X;\PiBone,Z\mid\PiAone)=0.
	$
	The remaining two terms are bounded by Alice's message lengths, and hence
	\begin{equation}\label{eq:alice-total-info}
		I(X;\Piqry,\PiAtwo)
		\leq H(\PiAone)+H(\PiAtwo)
		\leq |\PiAone|+|\PiAtwo|.
	\end{equation}

	By \eqref{eq:uniform-coordinate-info} and
	\eqref{eq:alice-total-info}, the first term in the right-hand side of \eqref{eq:alice-info-reveal} is at most
	$
	\frac{1}{m}(|\PiAone|+|\PiAtwo|),
	$
	which, combined with \eqref{eq:TV}, gives the lemma.
\end{proof}

\begin{proof}[Proof of Theorem~\ref{thm:tsi-lb}]
	The output of the augmented protocol predicts $X_J$ from
	$(Y,\Piqry,\PiAtwo)$ with error at most $\delta$.  Fano's
	inequality gives
	$$
	H(X_J\mid Y,\Piqry,\PiAtwo)\leq h_2(\delta).
	$$
	Combining this with Lemma~\ref{lem:alice-target-information} gives
	\begin{equation*}\label{eq:communication-tradeoff}
		1-h_2(\delta)
		\leq
		\frac{|\PiAone|+|\PiAtwo|}{m}
		+\sqrt{\frac{\ln2}{2}\,\frac{|\PiBone|}{m}}.
	\end{equation*}
	Since constant $\delta<1/2$, the left-hand side is a positive constant.  If the total communication were $o(m)$, then $|\PiAone|+|\PiAtwo|=o(m)$ and $|\PiBone|=o(m)$, so the right-hand side is $o(1)$, a contradiction.  Thus, every deterministic protocol having error at most $\delta$ under $\cD_{\mathrm{TSI}}$ uses
	$\Omega(m)$ bits.
\end{proof}

\subsection{Reduction from Two-Step Indexing to Distributed Diversity Index}
\label{sec:reduction-TSI-DDI}

We prove Theorem~\ref{thm:ddi-lb} by reducing $\TSI_m$ to $\DDI_m$.  The following lemma gives the degree of each node in the similarity graph under distribution $\cD_{\mathrm{DDI}}$.

\begin{lemma}
	\label{lem:degree-contributions}
	For each $i\in[m]$, let
	$
	n_i\coloneqq|\{q\in[m]:Y_q=i\}|.
	$
	The four node classes satisfy the following bounds:
	\begin{itemize}[leftmargin=2em]
		\item Every Alice's node $a_i$ has degree
		$
		d(a_i)=m+\gamma n_iX_i,
		$
		and the total contribution of the Alice's nodes is
		$
		\sum_{i\in[m]}\frac1{d(a_i)}\leq1.
		$
		\smallskip
		
		\item Every Bob's node $b_{Z,r}$ (i.e., $q = Z$) has degree
		$
		d(b_{Z,r})=1+X_{Y_Z},
		$
		and these nodes contribute
		$
		\sum_{r\in[\gamma]}\frac1{d(b_{Z,r})}
		=\frac{\gamma}{1+X_{Y_Z}}.
		$
		\smallskip
		
		\item For every $q\neq Z$, every Bob's node $b_{q,r}$ has degree
		$
		d(b_{q,r})=\gamma m+1+X_{Y_q},
		$
		and all these nodes together contribute
		$
		\sum_{q\in[m]\setminus\{Z\}}\sum_{r\in[\gamma]}
		\frac1{d(b_{q,r})}<1.
		$
		\smallskip
		
		\item Every Charlie's node $c_\ell$ has degree
		$
		d(c_\ell)=\gamma(2m-1),
		$
		and the total contribution of Charlie's nodes is
		$
		\sum_{\ell\in[\gamma m]}\frac1{d(c_\ell)}\leq1.
		$
	\end{itemize}
\end{lemma}

\begin{proof}
	The $m$ Alice's nodes form a clique, including their self-loops.  If
	$X_i=1$, then $a_i$ is also adjacent to all $\gamma n_i$ Bob's nodes whose
	$Y_q$ value is $i$.  Therefore
	$d(a_i)=m+\gamma n_iX_i\geq m$, and summing $1/d(a_i)$ over $i\in[m]$
	gives the first item.
	
	For every $r\in[\gamma]$, the node $b_{Z,r}$ has no Charlie-node neighbor.  It has its self-loop and is adjacent to $a_{Y_Z}$ exactly when
	$X_{Y_Z}=1$.  This gives the second item.
	
	If $q\neq Z$, then $b_{q,r}$ is adjacent to all $\gamma m$ Charlie's nodes, to itself, and possibly to $a_{Y_q}$.  Therefore $d(b_{q,r})=\gamma m+1+X_{Y_q}$, and
	$$
	\sum_{q\in[m]\setminus\{Z\}}\sum_{r\in[\gamma]}
	\frac1{d(b_{q,r})} \leq \frac{\gamma(m-1)}{\gamma m+1}<1.
	$$
	This gives the third item.
	
	Finally, the $\gamma m$ Charlie's nodes form a clique and are adjacent to all
	$\gamma(m-1)$ Bob's nodes with $q \neq Z$.  Hence
	$d(c_\ell)=\gamma m+\gamma(m-1)=\gamma(2m-1)$ and
	$$
	\sum_{\ell\in[\gamma m]}\frac1{d(c_\ell)}
	= \frac{\gamma m}{\gamma(2m-1)} = \frac{m}{2m-1}\leq1.
	$$
	This gives the fourth item.
\end{proof}

\begin{claim}
	\label{claim:di-gap}
	For the similarity graph induced by $(X,Y,Z)$, if $X_{Y_Z}=0$, then
	$
	\DI\geq\gamma.
	$
	If $X_{Y_Z}=1$, then
	$
	\DI\leq\frac{\gamma}{2}+3.
	$
	In particular, for $\gamma=20$, we have $\DI\geq20$ when
	$X_{Y_Z}=0$ and $\DI\leq13$ when $X_{Y_Z}=1$.
\end{claim}

\begin{proof}
	If $X_{Y_Z}=0$, the $\gamma$ nodes $b_{Z, r}$ already
	contribute $\gamma$.  If $X_{Y_Z}=1$, summing the contributions of the four groups of nodes in Lemma~\ref{lem:degree-contributions} gives
	$
	\DI\leq1+\frac{\gamma}{2}+1+1
	=\frac{\gamma}{2}+3.
	$
\end{proof}

\begin{proof}[Proof of Theorem~\ref{thm:ddi-lb}]
	Given an input $(X,Y,Z)$ to $\TSI_m$, Alice, Bob, and Charlie form the node sets $V_A(X)$, $V_B(Y)$, and $V_C(Z)$ from
	Definition~\ref{def:Dddi}.  They run the protocol for $\DDI$ on
	these sets and use its estimate to recover $X_{Y_Z}$.
	
	By Claim~\ref{claim:di-gap}, whenever the $\DDI$ protocol
	succeeds, if $X_{Y_Z}=0$, then
	$
	\widetilde D\geq0.9\cdot20=18.
	$
	If $X_{Y_Z}=1$, then
	$
	\widetilde D\leq1.1\cdot13=14.3.
	$
	Thus the threshold test $\widetilde D>16$ outputs $0$ in the first case
	and $1$ in the second.  A $0.49$-error $1.1$-approximation protocol for
	$\DDI$ would therefore give a $0.49$-error protocol for
	$\TSI_m$ with the same communication cost.  Theorem~\ref{thm:tsi-lb}
	gives an $\Omega(m)$ communication lower bound.
\end{proof}

\section{$L_{\DI}$-Sampling}
\label{sec:L-DI-sampling}

In this section, we give a three-pass $L_{\DI}$-sampler and then complement it with lower bounds.

\subsection{A Three-Pass Perfect $L_{\DI}$-Sampler}
\label{sec:L-DI-upper-bound}

For an array $A[1..m]$, we also use $A$ to denote the corresponding multiset.
For a node $u$ and a multiset $W$, let
$
d_W(u)=\sum_{v\in W}sim(v,u).
$
Let $c>0$ be any fixed constant specifying the additive error of the perfect sampler.

\vspace{2mm}
\noindent{\bf The algorithm and intuition.\ }
Due to its length, we present our $L_{\DI}$-sampling algorithm in two parts, Algorithms~\ref{alg:L-DI-preprocess} and~\ref{alg:L-DI-output}, with inline comments explaining the main steps.

\begin{algorithm}[!ht]
	\caption{First two passes of $L_{\DI}$-Sampling}
	\label{alg:L-DI-preprocess}
	\DontPrintSemicolon
	\SetAlgoNoEnd
	
	\KwIn{a node-arrival stream $\sigma=(\sigma_1,\ldots,\sigma_n)$}
	\KwOut{the state used by Algorithm~\ref{alg:L-DI-output}, or $\Null$}
	
	$w\gets3(c+1)\sqrt n\ln n$, $q\gets16(c+3)\sqrt n\ln n$, $\theta\gets6(c+1)\ln n$
	
	$W[1..w]$, $Q[1..q]\gets\Null$
	\gray{\tcc*{independent uniform samples}}
	
	$S\gets\emptyset$, $n_L\gets0$
	\gray{\tcc*{$S$ stores tuples (sampler index, node, degree)}}
	
	\nonl\gray{\tcp{Pass 1: obtain the witness set and estimate the number of low-degree nodes}}
	\ForEach{incoming $\sigma_j$}{
		\For{$a\gets1,\ldots,w$}{
			with probability $1/j$, set $W[a]\gets\sigma_j$
		}
		\For{$a\gets1,\ldots,q$}{
			with probability $1/j$, set $Q[a]\gets\sigma_j$
		}
	}
	
	$\widetilde n_L\gets\frac{n}{q}\abs{\{a\in[q]:d_W(Q[a])<\theta\}}$
	\gray{\tcc*{estimate the number of low-degree nodes}}
	
	\eIf{$\widetilde n_L\le\sqrt n$}{
		$\mathsf{case}\gets0$, $t\gets\lceil16cn\ln n\rceil$
	}{
		$\mathsf{case}\gets1$,
		$t\gets\left\lceil16cn\sqrt n\ln n/\widetilde n_L\right\rceil$
	}
	
	
	\nonl\gray{\tcp{Pass 2: apply the first rejection step and store surviving nodes}}
	Choose a random map $h:[t]\to[n]$
	
	\ForEach{incoming $\sigma_j$}{
		$\mathsf{low}_j\gets\one[d_W(\sigma_j)<\theta]$
		
		\If{$\mathsf{low}_j=1$}{
			$n_L\gets n_L+1$ 
		}
		
		\ForEach{$i\in[t]$ such that $h(i)=j$}{
			\leIf{$\mathsf{low}_j=1$}{
				$p_1\gets1$
			}{
				$p_1\gets1/\sqrt n$
			}
			\WithProb{$p_1$}{
				append $(i,\sigma_j,0)$ to $S$
				
				\If{$\abs S>80c\sqrt n\ln n$}{
					\Return $\Null$  \label{ln:space-cap-1}
					\gray{\tcc*{the space cap is exceeded}}
				}
			}
		}
	}
	
	\If{$\mathsf{case}=1$ and $n_L\le\sqrt n/2$}{
		\Return $\Null$ \label{ln:consistency-check-1}
		\gray{\tcc*{the first-pass estimate was too large}}
	}
	Sort the tuples in $S$ by their sampler indices
	
	\Return $(W,S,\theta)$
\end{algorithm}

\begin{algorithm}[!ht]
	\caption{Third pass and output for $L_{\DI}$-Sampling}
	\label{alg:L-DI-output}
	\DontPrintSemicolon
	\SetAlgoNoEnd
	
	\KwIn{the state $(W,S,\theta)$ returned by Algorithm~\ref{alg:L-DI-preprocess}}
	\KwOut{an $L_{\DI}$ sample, or $\Null$}
	
	\nonl\gray{\tcp{Pass 3: compute exact degrees}}
	\ForEach{incoming $\sigma_j$}{
		\ForEach{$(i,v,D)\in S$}{
			$D\gets D+sim(\sigma_j,v)$
		}
	}
	
	\nonl\gray{\tcp{Output: apply the final rejection step in sampler order}}
	\ForEach{$(i,v,D)\in S$ in increasing order of $i$}{
		\eIf{$d_W(v)\ge\theta$}{
			\If{$D<\sqrt n$}{
				\Return $\Null$  \label{ln:consistency-check-2}
				\gray{\tcc*{the high-low classification is inconsistent}}
			}
			$p_2\gets\sqrt n/D$
		}{
			$p_2\gets1/D$
		}
		with probability $p_2$, \Return $v$
	}
	\Return $\Null$
\end{algorithm}

We give a high-level description of the algorithm. The parameter choices in this overview are for illustrative purposes only; throughout, we omit logarithmic and constant factors.

The first pass of the algorithm samples a witness set $W$ of roughly $\sqrt{n}$ nodes. This set is used to classify nodes according to their degrees: a node is high-degree if its degree is at least $\sqrt{n}$, and low-degree otherwise.  

We note that this high-low partition of the nodes was also used in the $L_{\DI}$-sampling algorithm of prior work~\cite{LZ26b}. However, the algorithm of~\cite{LZ26b} performs $L_{\DI}$-sampling separately on the high-degree and low-degree nodes, and then merges the two resulting samples. The difficulty is that the merge step must choose between the two samples with probability proportional to the high- and low-degree contributions to $\DI(\sigma)$, which in turn requires separate $(1+\eps)$-approximations to these two quantities. As a result, the previous $L_{\DI}$-sampler incurs a $(1+\eps)$ multiplicative error and is therefore not perfect.

In its second pass, our $L_{\DI}$-sampling algorithm maintains a collection of \emph{conceptual samplers}. Each conceptual sampler samples a node uniformly at random from $\{\sigma_1,\ldots,\sigma_n\}$. Before actually storing the sampled node $\sigma_j$, the algorithm applies an immediate rejection step with probability $1/\sqrt{n}$ if $\sigma_j$ is high-degree. This rejection step is safe because the final desired sampling probability for $\sigma_j$ is $1/d_j \le 1/\sqrt{n}$. We do not apply this immediate rejection step to low-degree nodes, since for them $1/d_j$ may exceed $1/\sqrt{n}$. This immediate rejection sampling step is the key to our algorithm, as the sublinear-space constraint prevents us from explicitly materializing all conceptual samplers in memory. 

In the third pass, the algorithm performs an additional rejection step on each stored node, ensuring that the overall sampling probability of every stored node $\sigma_j$ is proportional to $1/d_j$. At the end, we select the first successful conceptual sampler and output the node it sampled.

It remains to choose the number of conceptual samplers to maintain. Let $t$ denote this number. If $t$ is too large, the space usage becomes too high; if $t$ is too small, the algorithm may end up with no sampled node. We therefore need to choose $t$ small enough to keep the space usage low, while still ensuring that at least one node survives until the end of the third pass.

High-degree nodes are not an issue, because the immediate rejection step ensures that the number of stored high-degree nodes is small. The main challenge is controlling the number of stored low-degree nodes. One could use an additional pass to determine the number of low-degree nodes and then choose $t$ accordingly, but this would increase the total number of passes to four. 

We instead use the witness set sampled in the first pass to estimate the number of low-degree nodes. Let this estimate be denoted by $\tilde{n}_L$. We then set the number of conceptual samplers based on this estimate. Roughly speaking, it is sufficient to set
$t = \min\left\{n, {n\sqrt{n}}/{\tilde{n}_L}\right\}$ to ensure that at least one node survives the third-pass rejection step.

\vspace{2mm}
\noindent{\bf The analysis.\ }
Let
$H=\{j\in[n]:d_W(\sigma_j)\ge\theta\}$ and
$L=[n]\setminus H$,
and write $n_H=\abs H$ and $n_L=\abs L$.

Define $\cE_1$ to be the event that
$\min_{j\in H}d_j\ge\sqrt n$ and
$\max_{j\in L}d_j\le4\sqrt n$.

\begin{lemma}
	\label{lem:L-DI-partition}
	$
	\Pr[\cE_1]\ge1-{2}/{n^c}.
	$
\end{lemma}

\begin{proof}
	Fix $j\in[n]$ and let $X_j=d_W(\sigma_j)$.  Since the entries of $W$ are independent uniform stream samples, the random variables $sim(W[a],\sigma_j)$, $a\in[w]$, are independent and lie in
	$[0,1]$.  Moreover,
	$$
	X_j=\sum_{a\in[w]}sim(W[a],\sigma_j), \quad \text{and} \quad \E[X_j]=\frac{3(c+1)d_j\ln n}{\sqrt n}.
	$$
	If $d_j\ge4\sqrt n$, then $\E[X_j]\ge12(c+1)\ln n$ and
	$$
	\Pr[X_j<6(c+1)\ln n] \le\Pr[X_j\le\E[X_j]/2] \le e^{-\E[X_j]/8} \le n^{-(c+1)}.
	$$
	If $d_j\le\sqrt n$, then $\E[X_j]\le3(c+1)\ln n$.  A Chernoff bound gives
	$$
	\Pr[X_j\ge6(c+1)\ln n]\le e^{-(c+1)\ln n}=n^{-(c+1)}.
	$$
	A union bound over the $n$ nodes proves the lemma.
\end{proof}

For the correctness analysis, consider an idealized version of Algorithms~\ref{alg:L-DI-preprocess} and~\ref{alg:L-DI-output} that is identical to the actual algorithm except that it does not impose the space cap (Line~\ref{ln:space-cap-1} of Algorithm~\ref{alg:L-DI-preprocess}) and does not return $\Null$ at the two consistency checks (Line~\ref{ln:consistency-check-1} of Algorithm~\ref{alg:L-DI-preprocess} and Line~\ref{ln:consistency-check-2} of Algorithm~\ref{alg:L-DI-output}). We first establish the exact output distribution of this idealized algorithm, and then show that it differs from the actual algorithm only with small probability.

\begin{lemma}
	\label{lem:L-DI-exact-law}
	Fix a realization of $W$ satisfying $\cE_1$ and a realization of the independent sample $Q$; together, they determine the number $t$ of conceptual samplers.  For every $j\in[n]$ and every conceptual
	sampler $i$,
	$
	\Pr[\text{sampler $i$ succeeds with }\sigma_j\mid W,Q] = {1}/{(nd_j)}.
	$
	Consequently, conditioned on at least one conceptual sampler succeeding, the first successful sampler outputs $\sigma_j$ with probability	${d_j^{-1}}/{\DI}.$
\end{lemma}

\begin{proof}
	The $i$-th conceptual sampler selects $\sigma_j$ with probability $1/n$.  If
	$j\in H$, the two rejection probabilities multiply to
	$
	\frac1{\sqrt n}\cdot\frac{\sqrt n}{d_j}=\frac1{d_j}.
	$
	If $j\in L$, they multiply to $1\cdot \frac1{d_j}=\frac1{d_j}$.  This proves the first
	claim.
	
	Each conceptual sampler therefore succeeds with probability $\DI/n$.  By
	independence, the probability that the first successful sampler is sampler
	$i$ and outputs $\sigma_j$ is
	$
	\left(1-\frac{\DI}{n}\right)^{i-1}\frac1{nd_j}.
	$
	Summing over $i\in[t]$ gives
	$$
	\Pr[\text{the first successful sampler outputs }\sigma_j] = \frac{1}{nd_j} \sum_{i \in [t]} \left(1-\frac{\DI}{n}\right)^{i-1}.
	$$
	Similarly, since each sampler succeeds with probability $\DI/n$,
	$$
	\Pr[\text{at least one sampler succeeds}] = \frac{\DI}{n}\sum_{i \in [t]} \left(1-\frac{\DI}{n}\right)^{i-1}.
	$$
	Therefore,
	\begin{equation*}
		\Pr[\text{output }\sigma_j
		\mid \text{at least one sampler succeeds}]=
		\frac{
			\frac{1}{nd_j}\sum_{i \in [t]}\left(1-\frac{\DI}{n}\right)^{i-1}
		}{
			\frac{\DI}{n}\sum_{i \in [t]}\left(1-\frac{\DI}{n}\right)^{i-1}
		}
		= \frac{d_j^{-1}}{\DI}.
	\end{equation*}
	Thus, conditioned on at least one conceptual sampler succeeding, the output
	follows the exact $L_{\DI}$ distribution.
\end{proof}

The next lemma bounds the storage and failure events.  

Let $\cE_Q$ be the following event: if $n_L\le\sqrt n/2$, then $\widetilde n_L\le\sqrt n$; if $n_L>\sqrt n/2$, then $n_L/2\le\widetilde n_L\le2n_L$.

\begin{lemma}
	\label{lem:L-DI-resources}
	Conditioned on any fixed $W$, we have
	$
	\Pr[\cE_Q\mid W]\ge1-{2}/{n^{c+3}}.
	$
	Conditioned on $\cE_1\cap\cE_Q$, the probability that the algorithm exceeds its space cap is $e^{-\Omega(\sqrt n\ln n)}$, and the probability that none of
	the conceptual samplers succeeds is at most
	$
	\left(1/{n^c}+e^{-\Omega(\sqrt n\ln n)}\right).
	$
\end{lemma}

\begin{proof}
	Let $Y=\abs{\{a\in[q]:Q[a]\in L\}}.$  Note that $q$ is set to be $16(c+3)\sqrt{n}\ln n$ in the algorithm. Conditioned on $W$, the random variable $Y$ is binomial with mean
	$\mu_Q=q\frac{n_L}{n}=16(c+3)\frac{n_L}{\sqrt n}\ln n$, and $\widetilde n_L=\frac nqY.$ If $n_L\le\sqrt n/2$, then $\mu_Q\le8(c+3)\ln n$.  By stochastic domination and a Chernoff bound,
	$$
	\Pr[\widetilde n_L>\sqrt n\mid W] = \Pr[Y>16(c+3)\ln n\mid W] \le e^{-8(c+3)\ln n/3} \le n^{-(c+3)}.
	$$
	If $n_L>\sqrt n/2$, then $\mu_Q>8(c+3)\ln n$, and
	$$
	\Pr\left[Y<\frac{\mu_Q}{2}\ \text{or}\ Y>2\mu_Q\mid W\right] \le e^{-\mu_Q/8}+e^{-\mu_Q/3} \le 2n^{-(c+3)}.
	$$
	This proves the first item.
	
	Let $X$ be the number of nodes stored after the first rejection step.  Since the
	conceptual samplers are independent, $X$ is binomial.  If
	$n_L\le\sqrt n/2$ and $\cE_Q$ holds, then $t=16cn\ln n$ and
	$$
	\E[X] = t\left(\frac{n_L}{n}+\frac{n_H}{n\sqrt n}\right) \le 24c\sqrt n\ln n.
	$$
	
	Suppose now that $n_L>\sqrt n/2$, $\cE_Q$ holds, and $\widetilde n_L>\sqrt n$. Then $t=16cn\sqrt n\ln n/\widetilde n_L$, and hence $t n_L/n\le32c\sqrt n\ln n$ and $t n_H/(n\sqrt n)<16c\sqrt n\ln n$.
	If $\widetilde n_L\le\sqrt n$, then $t=16cn\ln n$ and
	$n_L\le2\widetilde n_L\le2\sqrt n$, which gives the same two upper bounds.
	In both cases, the expected numbers of stored low- and high-degree nodes are at most $32c\sqrt n\ln n$ and $16c\sqrt n\ln n$, respectively. Therefore, $\E[X] \le 48c\sqrt n\ln n.$
	
	Since $X$ is a sum of independent Bernoulli random variables, a Chernoff bound gives
	$$
	\Pr\left[X>80c\sqrt n\ln n\right] \le \exp \left(-\Omega(\sqrt n\ln n)\right) \le n^{-c}.
	$$
	Hence, with probability at least $1-1/n^c$, the space cap is not exceeded.
	
	It remains to show that some conceptual sampler succeeds.  If $n_L\le\sqrt n/2$, then $t=16cn\ln n$.  By
	Lemma~\ref{lem:L-DI-exact-law}, each conceptual sampler succeeds with probability $\DI/n$ conditioned on $W$ and $Q$.  Since $\DI\ge1$,
	$$
	\Pr[\text{no success}\mid\cE_1,\cE_Q] = \left(1-\frac{\DI}{n}\right)^t \le e^{-t/n} \le n^{-16c}.
	$$
	
	Now assume $n_L>\sqrt n/2$.  Let $X_L$ be the number of low-degree conceptual samples stored
	after Pass~2.  If $\widetilde n_L>\sqrt n$, then
	$$
	\E[X_L] = \frac{16cn\sqrt n\ln n}{\widetilde n_L} \cdot \frac{n_L}{n} \ge 8c\sqrt n\ln n.
	$$
	If $\widetilde n_L\le\sqrt n$, then
	$$
	\E[X_L]=16cn_L\ln n>8c\sqrt n\ln n.
	$$
	Hence, a Chernoff bound gives
	$$
	\Pr[X_L<4c\sqrt n\ln n]\le e^{-c\sqrt n\ln n}.
	$$
	
	Given $\cE_1$, every stored low-degree sample has degree at most $4\sqrt n$, so its final rejection step succeeds with probability at least
	$1/(4\sqrt n)$.  Conditional on $X_L\ge4c\sqrt n\ln n$, these rejection steps are independent, and therefore
	$$
	\Pr[\text{no low-degree node succeeds}] \le \left(1-\frac1{4\sqrt n}\right)^{4c\sqrt n\ln n} \le n^{-c}.
	$$
	Combining the two cases proves the lemma.
\end{proof}

\begin{theorem}
	\label{thm:L-DI-upper}
	For every symmetric similarity function
	$sim:\cU\times\cU\to[0,1]$ satisfying
	$sim(x,x)=1$ for every $x\in\cU$, and every fixed constant $c>0$, Algorithms~\ref{alg:L-DI-preprocess} and~\ref{alg:L-DI-output} form an
	$O(n^{-c})$-perfect $L_{\DI}$-sampler.  They
	use three passes and $O_c(\sqrt n\ln n)$ words of space.
\end{theorem}

\begin{proof}
	Let $\cB$ be the event that $\cE_1$ or $\cE_Q$ fails, or that the number of stored nodes exceeds the space cap.  By Lemmas~\ref{lem:L-DI-partition} and~\ref{lem:L-DI-resources}, $\Pr[\cB]=O(n^{-c}).$ When $\cB$ does not occur, neither consistency check is triggered, and Algorithms~\ref{alg:L-DI-preprocess} and~\ref{alg:L-DI-output} behave exactly as the idealized algorithm. Moreover, Lemma~\ref{lem:L-DI-resources} shows that the probability that no conceptual sampler succeeds is $O(n^{-c})$.  Hence, the algorithm returns $\Null$ with probability $O(n^{-c})$.
	
	It remains to bound the error in the conditional output distribution.  Let $\pi=(d_j^{-1}/\DI)_{j\in[n]}$, and let $X$ denote the output of the algorithm.  By Lemma~\ref{lem:L-DI-exact-law}, for every fixed $W$
	satisfying $\cE_1$ and every fixed $Q$, conditioned on success, the idealized algorithm has distribution $\pi$.  Since $\pi$ is independent of $W$ and $Q$, the same remains true after averaging over these random choices.
	
	Let $r$ be the probability that $\cE_1$ holds and the idealized algorithm returns a node.  For every set $A$ of stream nodes,
	$$
	\Pr[\text{idealized output}\in A,\cE_1]=r\cdot\pi(A).
	$$
	Since the two executions agree whenever $\cB$ does not occur, $\left|\Pr[X\in A]-r \cdot \pi(A) \right|\le\Pr[\cB]$ and $\left|\Pr[X\ne\Null]-r\right|\le\Pr[\cB]$.
	Therefore,
	$$
	\left|\Pr[X\in A\mid X\ne\Null]-\pi(A)\right| \le \frac{2\Pr[\cB]}{\Pr[X\ne\Null]} = O(n^{-c}).
	$$
	Taking the supremum over all sets $A$ of stream nodes and using the definition of total variation distance gives that, conditioned on success, the output distribution is within $O(n^{-c})$ total variation distance of the exact $L_{\DI}$ distribution.
	
	Finally, the arrays $W$ and $Q$ use $O_c(\sqrt n\ln n)$ words, and the explicit cap on $S$ bounds the space used by the conceptual samplers by the same quantity.  All other data structures use lower-order space. 	The algorithm makes exactly three passes.  
	
	We note that $h$ is not stored explicitly: we set $h(i)\coloneqq G(\rho,i)$, where $G$ is Nisan's generator~\cite{Nisan90} with an $O_c(\log^2 n)$-bit
	seed $\rho$, and recompute $h(i)$ from $\rho$ whenever needed.  This does not affect the analysis.  Fix the stream, $W$, and $Q$, and let $\zeta_i$ and $\xi_i$ denote the random values used by sampler $i$ in the two rejection steps.  The output of the algorithm is a function of $(h(i),\zeta_i,\xi_i)_{i\in[t]}$ computable by a program that reads this sequence once, in the order $i=1,\ldots,t$, with $O(\log n)$ bits of state:
	upon reading the $i$-th block it decides whether sampler $i$ survives both rejection steps, updates a counter for $|S|$, and records the node if this
	is the first success.  Nisan's generator fools ever such program to within $n^{-(c+1)}$ in total variation distance~\cite{Nisan90,Ind06}, so every probability in
	Lemmas~\ref{lem:L-DI-exact-law}--\ref{lem:L-DI-resources} changes by at most $n^{-(c+1)}$, which is absorbed into the $O(n^{-c})$ bounds above.

\end{proof}

\subsection{A Lower Bound for Perfect $L_{\DI}$-Sampling}
\label{sec:L-DI-lower-bound}

We reduce diversity-index estimation to perfect $L_{\DI}$-sampling.  The algorithm is presented in Algorithm~\ref{alg:DI-from-L-DI}.


\begin{algorithm}[!t]
	\caption{$\DI$-Estimation from Perfect $L_{\DI}$-Sampling}
	\label{alg:DI-from-L-DI}
	\DontPrintSemicolon
	\SetAlgoNoEnd
	
	\KwIn{a stream $\sigma$ of $n$ nodes, a perfect $L_{\DI}$-sampler $\mathcal A$, and $\eps\in(0,1/5)$}
	\KwOut{a $(1+\eps)$-approximation of $\DI(\sigma)$, or $\Null$}
	
	$b\gets2^{16}\eps^{-2}\ln n$
	
	\For{$\ell\gets0,1,\ldots,\lfloor\log_2n\rfloor$}{
		let $\sigma^{(\ell)}$ be $\sigma$ followed by $2^\ell$ pairwise isolated dummy nodes
		
		run $b$ independent copies of $\mathcal A$ on $\sigma^{(\ell)}$
		
		$Y^{(\ell)}\gets$ number of non-null outputs
		
		$X^{(\ell)}\gets$ number of dummy-node outputs
		
		\If{$Y^{(\ell)}\ge b/3$ and
			$X^{(\ell)}/Y^{(\ell)}\in[1/4,5/8]$}{
			\Return $\left(Y^{(\ell)}/X^{(\ell)}-1\right)2^\ell$
		}
	}
	\Return $\Null$
\end{algorithm}

All copies in Algorithm~\ref{alg:DI-from-L-DI}, for all values of $\ell$, are run in parallel.  Each original stream node is forwarded to every copy, and the dummy nodes are appended to each virtual stream at the end
of every pass.  The dummy nodes are similar only to themselves.  Therefore,
each dummy node has degree one and
$\DI(\sigma^{(\ell)})=\DI(\sigma)+2^\ell.$

\begin{lemma}
	\label{lem:L-DI-reduction}
	Suppose $\mathcal A$ is a $\kappa$-pass, $s$-space perfect $L_{\DI}$-sampler with failure probability at most $0.49$.  Then, for every fixed constant $\eps\in(0,1/5)$, Algorithm~\ref{alg:DI-from-L-DI} is a $\kappa$-pass $(1+\eps)$-approximation algorithm for $\DI$ with failure probability $o(1)$ and space $O\bigl(s\eps^{-2}\log^2n\bigr).$
\end{lemma}

\begin{proof}
	Consider a fixed $\ell$, set $m_\ell:= 2^\ell$
	and $q_\ell :={m_\ell}/({m_\ell+\DI}).$ Since every dummy node has degree one, $q_\ell$ is the probability that an exact $L_{\DI}$ sample from $\sigma^{(\ell)}$ is a dummy node. Let $p_\ell$
	be the corresponding probability conditioned on $\mathcal A$ returning a node. Since $\mathcal A$ is $n^{-c}$-perfect and $\sigma^{(\ell)}$ has $O(n)$ nodes, we have $|p_\ell-q_\ell|\le\gamma_\ell$, where
	$\gamma_\ell=O(n^{-c})$. For a fixed constant $\eps$, we have $\gamma_\ell\le\eps/64$ simultaneously for all $\ell$.
	
	Let $s_\ell$ be the success probability of one copy of $\mathcal A$ on
	$\sigma^{(\ell)}$. We have $s_\ell\ge0.51$.  Moreover,
	$$
	Y^{(\ell)}\sim\operatorname{Binomial}(b,s_\ell) \quad \text{and} \quad \left(X^{(\ell)}\mid Y^{(\ell)}=y\right) \sim\operatorname{Binomial}(y,p_\ell).
	$$
	By Hoeffding's inequality and the choice of $b$, for each fixed $\ell$,
	\begin{equation}
		\Pr\Bigg[\left(Y^{(\ell)}<\frac b3\right) \vee \left(\left|
		\frac{X^{(\ell)}}{Y^{(\ell)}}-p_\ell \right|>\frac{\eps}{64} \right)\Bigg] \le n^{-10} + 2\exp\left(-\frac{b\eps^2}{6144}\right)
		\le 3n^{-10},
	\end{equation}
	where the ratio $X^{(\ell)}/Y^{(\ell)}$ is considered only when $Y^{(\ell)}\ge b/3$. A union bound
	over the $O(\log n)$ choices of $\ell$ gives that, with probability at least
	$1-n^{-9}$, simultaneously for every $\ell$,
	\begin{equation}
		Y^{(\ell)}\ge\frac b3 \qquad\text{and}\qquad
		\left|\frac{X^{(\ell)}}{Y^{(\ell)}}-q_\ell\right| \le\frac{\eps}{32}.
		\label{eq:empirical-vs-exact-dummy}
	\end{equation}
	We condition on this event.
	
	Suppose the algorithm returns at scale $\ell$, and write
	$
	\widehat q_\ell={X^{(\ell)}}/{Y^{(\ell)}}.
	$
	The returned estimate $\widehat{\DI}$ and the true value of $\DI$ satisfy
	$
	\widehat{\DI} = m_\ell\frac{1-\widehat q_\ell}{\widehat q_\ell}
	$ 
	and $\DI = m_\ell\frac{1-q_\ell}{q_\ell}$, respectively.
	Therefore,
	\begin{equation}
		\left| \frac{\widehat{\DI}}{\DI}-1 \right| = \left| \frac{q_\ell(1-\widehat q_\ell)}{\widehat q_\ell(1-q_\ell)}-1 \right| = \frac{|q_\ell-\widehat q_\ell|}{\widehat q_\ell(1-q_\ell)}.
		\label{eq:DI-relative-error}
	\end{equation}
	Because scale $\ell$ is accepted, $\widehat q_\ell\in[\frac{1}{4},\frac{5}{8}]$. Moreover,
	\eqref{eq:empirical-vs-exact-dummy} and $\eps\le1/5$ imply 
	$
	q_\ell \le \frac58+\frac{\eps}{32} \le \frac{101}{160}
	$
	and
	$
	\widehat q_\ell(1-q_\ell) \ge\frac14\cdot\frac{59}{160} = \frac{59}{640}.
	$
	Substituting these bounds into \eqref{eq:DI-relative-error} gives
	$
	\left|\frac{\widehat{\DI}}{\DI}-1\right| \le \frac{\eps/32}{59/640} < \eps.
	$
	Thus, every returned value is a $(1+\eps)$-approximation to $\DI$.
	
	It remains to show that the algorithm returns a value. Let $\ell^*=\lfloor\log \DI\rfloor$. Since $1\le\DI\le n$, this is one of the scales considered by the algorithm, and it holds that $m_{\ell^*}\le\DI<2m_{\ell^*}$ and $\frac13<q_{\ell^*}\le\frac12$. By \eqref{eq:empirical-vs-exact-dummy},
	$$
		\frac13-\frac{\eps}{32} < \frac{X^{(\ell^*)}}{Y^{(\ell^*)}} \le \frac12+\frac{\eps}{32}.
	$$
	For $\eps\le1/5$, this interval is contained in $[1/4,5/8]$. Hence, scale $\ell^*$ is accepted unless an earlier scale has already been accepted, and
	in either case the algorithm returns a $(1+\eps)$-approximation to $\DI$.
	
	Finally, the algorithm considers $O(\log n)$ scales and runs $b=O(\eps^{-2}\log n)$ copies of $\mathcal A$ at each scale. All copies are run in parallel, so the number of passes remains $\kappa$, and the total
	space usage is $O\left(s\eps^{-2}\log^2 n\right).$
\end{proof}

\begin{theorem}
	\label{thm:L-DI-lb}
	Any two-pass perfect $L_{\DI}$-sampler with failure probability at most $0.49$ requires
	$\Omega\left({n}/{\log^2n}\right)$ bits of space.
\end{theorem}

\begin{proof}
	Suppose such a sampler uses $s(N)$ bits on streams of length at most $N$.  Apply Lemma~\ref{lem:L-DI-reduction} with
	$\eps=0.1$ to an $n$-node input.  Every virtual stream has length at most $2n$, so the resulting two-pass diversity-index estimator uses $O(s(2n)\log^2n)$ bits and succeeds with probability greater than $0.51$.  Theorem~\ref{thm:streaming-main} gives
	$s(2n)\log^2n=\Omega(n).$ Replacing $2n$ by $N$ proves the theorem.
\end{proof}

For a constant number of passes, we use the following lower bound from previous work.  For that, we first need to define approximate sampler.

\begin{definition}[Approximate sampler]
	\label{def:approximate-sampling}
	Let $\pi$ be a target distribution on the stream nodes.  A streaming algorithm is a $(1+\eps)$-approximate sampler for $\pi$ with failure probability $\lambda$ if it outputs a stream node $J$ or $\FAIL$, and $\Pr[J=\FAIL]\le\lambda$ and $\Pr[J=j\mid J\ne\FAIL]=(1\pm\eps)\,\pi(j)$ for every $j \in [n]$.
\end{definition}

\begin{theorem}[\cite{LZ26b}]
	\label{thm:prior-L-DI-constant-pass-lb}
	Let $\eps\in(0,0.9)$.  Any constant-pass $(1+\eps)$-approximate $L_{\DI}$-sampler with failure probability at most $0.49$ requires $\Omega(\sqrt n)$ bits of space.
\end{theorem}

\begin{theorem}
	\label{thm:L-DI-constant-pass-lb}
	Fix a constant $c>2$ and an integer $\kappa\ge1$.  Any $\kappa$-pass $n^{-c}$-perfect $L_{\DI}$-sampler requires $\Omega(\sqrt n)$ bits of space.
\end{theorem}

\begin{proof}
	For every node $\sigma_j$, we have $d_j\le n$ and $\DI\le n$, and hence $d_j^{-1}/\DI\ge1/n^2.$
	Therefore, an $n^{-c}$-perfect sampler returns $\sigma_j$, conditioned on success, with probability
	$
		\frac{d_j^{-1}}{\DI}\pm n^{-c} = \left(1\pm n^{2-c}\right)\frac{d_j^{-1}}{\DI}.
	$
	Since $c>2$, this is a $(1+0.1)$-approximate $L_{\DI}$-sampler, and its failure probability is at
	most $n^{-c}<0.49$.  Applying
	Theorem~\ref{thm:prior-L-DI-constant-pass-lb} with $\eps=0.1$ proves the theorem.
\end{proof}

\section{A Two-Pass Perfect $L_{M_p}$-Sampler}
\label{sec:LMp-sampling}

Fix constants $p>0$ and $c>0$.  In this section, we give a two-pass
$n^{-c}$-perfect $L_{M_p}$-sampler using
$O_{p,c}(n^{p/(p+1)}\log n)$ words of space.  

For a multiset $W$ and a node $v$, let
$d_W(v)\coloneqq\sum_{x\in W}sim(x,v)$.  When
$|W|=w$, write $\wt d(v)=\frac nw d_W(v)$.  The algorithm uses $W$ as a witness set to classify nodes and to estimate high degrees up to a constant factor.

Before presenting the $L_{M_p}$-sampling algorithm, we provide some background and a brief technical overview.

\subsection{Background and Technical Overview}
We first explain why two natural approaches do not give the desired two-pass perfect sampler.  A direct adaptation of our three-pass $L_{\DI}$-sampling approach would require an additional pass to recover the exact degree of a high-degree node selected in Pass~2, and hence would not give a two-pass algorithm. 

The previous two-pass $L_{M_p}$-sampling algorithm~\cite{LZ26b} uses the following exact exponential-scaling identity.  Let $R_1,\ldots,R_n$ be independent $\operatorname{Exp}(1)$ random variables.  The node with the smallest value of $R_j/d_j^p$ is distributed exactly according to
$d_j^p/M_p$.  For low-degree nodes, the previous algorithm stores every node that can be the minimizer with high probability and computes its exact degree in Pass~2.  The difficulty lies in the high-degree nodes.  A high-degree node is processed only when it appears in Pass~2, and by then the prefix of the stream needed to compute its exact degree has already passed.  The previous algorithm therefore uses the witness estimate $\wt d_j$ and compares $R_j/\wt d_j^p$.  Conditioned
on the witness set, this gives the high-degree node $\sigma_j$ the weight $\wt d_j^p$, rather than the required weight $d_j^p$.  Unless all degree estimates are exact, the conditional output distribution is not the exact $L_{M_p}$ distribution.  A larger or more accurate witness set makes this distortion smaller, but does not remove it.  This is why the previous analysis gives a
$(1+\eps)$-approximate sampling guarantee rather than perfect sampling in the sense of Definition~\ref{def:perfect-sampling}.

\vspace{2mm}
\noindent{\bf The new exact degree correction.\ }
The new ingredient is a deferred exact correction for a high-degree node selected in Pass~2.  The witness estimate is used only to construct a degree
upper bound $U$ satisfying
$
4d(v)\le U\le16d(v).
$
The algorithm first selects the node with weight $U^p$.  It then accepts the sample with probability $(d(v)/U)^p$, so the extra weight cancels exactly:
$$
U^p\left(\frac{d(v)}{U}\right)^p=d(v)^p.
$$
The remaining question is how to perform this rejection-sampling step {\em without} knowing $d(v)$ and {\em without} making a third pass.

In Pass~1, the algorithm stores $O(\log n)$ independent degree-witness sets. For each set $P_s$, every stream node $x$ is assigned an independent random value $R_s(x)\sim\operatorname{Unif}[0,1]$, and we store $x$ whenever $R_s(x)\le1-e^{-1/(4\Delta)}$, where $\Delta\coloneqq n^{1/(p+1)}$. After a high-degree node $v$ and its value $U$ have been selected in Pass~2, we use $P_s$ to generate a bit that equals $1$ if some stream node $x$ satisfies $R_s(x)\le1-e^{-sim(x,v)/U}$, and equals $0$ otherwise. Since each stream node $x$ fails this test independently with probability $e^{-sim(x,v)/U}$, this bit equals $1$ with probability
$$
q=1-e^{-d(v)/U}.
$$
Moreover, $U\ge4\Delta$ and $sim(x,v)\le1$, so $1-e^{-sim(x,v)/U}\le1-e^{-1/(4\Delta)}$; hence every random value needed for this test was stored in Pass~1.

By the exact coin-simulation theorem of Nacu and Peres~\cite{NP05}, independent $\operatorname{Bernoulli}(q)$ samples can then be used to generate a coin with success probability
$$
\bigl(-\ln(1-q)\bigr)^p = \left(\frac{d(v)}{U}\right)^p.
$$
Thus, the random information stored in Pass~1 allows the algorithm to apply an exact degree-dependent correction to a node that is selected only in Pass~2.
\smallskip

Now, fix a witness set for which the high-low classification is correct, and let
$\cH$ and $\cL$ denote the sets of indices corresponding to nodes classified as high and low degree,
respectively.  For $j\in\cH$, let $U_j=U(\sigma_j)$, and define
$$
B_L\coloneqq n\Delta^p, \qquad B_H\coloneqq\sum_{j\in\cH}U_j^p, \qquad \text{and} \quad B\coloneqq B_L+B_H.
$$
In one trial, the algorithm chooses the low-degree procedure with probability
$B_L/B$ and the high-degree procedure with probability $B_H/B$.

For a low-degree node $\sigma_j$, the algorithm samples a node uniformly and
uses its exact degree in the rejection step.  Hence
$$
\Pr[\text{one trial returns }\sigma_j] = \frac{B_L}{B}\cdot\frac1n\cdot \left(\frac{d_j}{\Delta}\right)^p = \frac{d_j^p}{B}.
$$
For a high-degree node $\sigma_j$, weighted sampling selects it with
probability $U_j^p/B_H$, and the exact rejection step accepts it with
probability $(d_j/U_j)^p$.  Therefore,
$$
\Pr[\text{one trial returns }\sigma_j] =\frac{B_H}{B}\cdot\frac{U_j^p}{B_H}\cdot \left(\frac{d_j}{U_j}\right)^p =\frac{d_j^p}{B}.
$$
Thus every node has the same normalizing denominator $B$.  Conditioned on a
trial returning a node, the output probability of $\sigma_j$ is exactly
$d_j^p/M_p$.  Repeating independent trials and returning the first successful
one preserves this distribution.

\subsection{The Two-Pass Algorithm}

\noindent{\bf Exact rejection sampling from degree-witness sets.\ }
We use the following consequence of the exact coin-simulation theorem of Nacu and Peres~\cite{NP05}. A $q$-bit is a Bernoulli random variable that equals $1$ with probability $q$ and $0$ otherwise.  Uniformity of $a_p,\rho_p$ over $q\in I$ follows from Proposition~21 of~\cite{NP05}, as $I$ is a closed interval and $\phi_p$ is real analytic on $I$ with values in $(0,1)$.

\begin{lemma}[\cite{NP05}]
	\label{lem:exact-rejection}
	Let $I\coloneqq\left[1-e^{-1/16},\,1-e^{-1/4}\right]$
	and $\phi_p(q)\coloneqq\bigl(-\ln(1-q)\bigr)^p.$ For every fixed $p>0$, there is an algorithm $\ExactReject$ that, given
	independent $q$-bits for an unknown $q\in I$, outputs a bit $A$ satisfying $\Pr[A=1]=\phi_p(q).$  Moreover, there are constants $a_p>0$ and $\rho_p\in(0,1)$, depending only on $p$ and \emph{not} on $q$, such that the number $N$ of input bits used satisfies $\Pr[N>m]\le a_p\rho_p^m$ for
	every $q\in I$ and every integer $m\ge0$.
\end{lemma}

Our algorithm uses $K=\Theta_{p,c}(\log n)$ independent degree-witness sets
$P_1,\ldots,P_K$. For every $s\in[K]$ and every stream node $x$, we independently draw $R_s(x)\sim\operatorname{Unif}[0,1]$.  The set $P_s$ stores $(x,R_s(x))$ only when $R_s(x)\le\tau_0$, where $\tau_0\coloneqq1-e^{-1/(4\Delta)}$.

For a node $v$ and a value $U\ge4\Delta$, define the bit
\begin{equation*}
	\label{eq:degree-coin}
	\DegreeCoin(v,U;s) \coloneqq \one\left[\exists j\in[n]:\ R_s(\sigma_j)\le1-e^{-sim(\sigma_j,v)/U}\right].
\end{equation*}
The bit can be computed from $P_s$ because, for every stream node $x$,
$1-e^{-sim(x,v)/U}\le1-e^{-1/U}\le\tau_0$.  It equals zero exactly when every stream node $x$ receives a random value larger than $1-e^{-sim(x,v)/U}$.  Therefore,
\begin{equation}
	\label{eq:degree-coin-probability}
	\Pr[\DegreeCoin(v,U;s)=1]
	=1-\prod_{j\in[n]}e^{-sim(\sigma_j,v)/U}
	=1-e^{-d(v)/U}.
\end{equation}
The bits obtained from different sets $P_s$ are independent.

\vspace{2mm}
\noindent{\bf The algorithm.\ }
Let $a_p$ and $\rho_p$ be the constants in
Lemma~\ref{lem:exact-rejection}, and set $\alpha_p\coloneqq16^{-p}$.
Fix a sufficiently small constant
$\eta_p\in(0,-\ln\rho_p)$ such that
\begin{equation}
	\label{eq:lambda-choice}
	\frac{a_p(e^{\eta_p}-1)}{1-e^{\eta_p}\rho_p} \le \frac{\alpha_p}{2}.
\end{equation}
Such a constant exists since the left-hand side tends to zero as $\eta_p$ tends to zero.

The perfect $L_{M_p}$-sampling algorithm is presented in Algorithms~\ref{alg:preprocess} and~\ref{alg:output} with inline comments.
Algorithm~\ref{alg:preprocess} describes the two streaming passes, and
Algorithm~\ref{alg:output} describes the output procedure.  

The array $L[1..t]$ consists of independent uniform stream samples selected in Pass~1, and $D_L[r]$ records the exact degree of $L[r]$ computed in Pass~2.  During Pass~2, each $H[r]$ is maintained as an independent weighted sample from the nodes classified as high degree, and $B_H$ records their total weight.  The algorithm classifies a node $v$ as high degree if $d_W(v)\ge\theta$ and low degree otherwise.  For a high-degree node, define
\begin{equation}
	\label{eq:degree-upper-bound}
	U(v)\coloneqq\max\{4\Delta,8\wt d(v)\}.
\end{equation}

\begin{algorithm}[!ht]
	\caption{Two-pass preprocessing for $L_{M_p}$-sampling}
	\label{alg:preprocess}
	\DontPrintSemicolon
	\SetAlgoNoEnd
	
	\KwIn{a stream $\sigma=(\sigma_1,\ldots,\sigma_n)$}
	\KwOut{the state used by Algorithm~\ref{alg:output}, or $\FAIL$}
	
	$C_W\gets64(c+10)$,
	$w\gets C_W\frac{n}{\Delta}\ln n$, $\theta\gets\frac{C_W}{2}\ln n$,
	$t\gets2(c+10)(\Delta^p+16^p)\ln n$
	
	$C_K\gets\max\left\{\frac{12(c+8)}7,\frac{c+9}{\eta_p}\right\}$,
	$K\gets C_K\ln n$,
	$S_{\max}\gets2K\frac n\Delta$
	
	Sample independent uniform maps $h_W:[w]\to[n]$ and $h_L:[t]\to[n]$
	\gray{\tcc*{choose uniform stream positions}}
	
	Initialize $W[1..w]$, $L[1..t]$, and $P_1,\ldots,P_K$
	\gray{\tcc*{witness set, low-degree samples, and degree-witness sets}}
	
	\nonl\gray{\tcp{Pass 1: store witness set, low-degree samples, and degree-witness sets}}
	\ForEach{incoming $\sigma_j$}{
		\lForEach{$a\in[w]$ such that $h_W(a)=j$}{
			$W[a]\gets\sigma_j$
		}
		\lForEach{$r\in[t]$ such that $h_L(r)=j$}{
			$L[r]\gets\sigma_j$
		}
		\For{$s\gets1,\ldots,K$}{
			draw $R_s(\sigma_j) \sim \Unif[0,1]$
			
			\If{$R_s(\sigma_j)\le1-e^{-1/(4\Delta)}$}{
				insert $(\sigma_j,R_s(\sigma_j))$ into $P_s$
				\gray{\tcc*{store only values that may be queried later}}
				
				\If{the total number of entries in $P_1,\ldots,P_K$ exceeds $S_{\max}$}{
					\Return $\FAIL$
					\gray{\tcc*{the space cap is exceeded}}
				}
			}
		}
	}
	
	\nonl\gray{\tcp{Pass 2: compute exact low degrees and obtain weighted high-degree samples}}
	$D_L[1..t]\gets0$, $H[1..t]\gets\FAIL$, $B_H\gets0$
	\gray{\tcc*{$H[r]$ is the $r$th weighted high-degree sample}}
	
	\ForEach{incoming $\sigma_j$}{
		\For{$r\gets1,\ldots,t$}{
			$D_L[r]\gets D_L[r]+sim(\sigma_j,L[r])$
			\gray{\tcc*{compute the exact degree of $L[r]$}}
		}
		\If{$d_W(\sigma_j)\ge\theta$}{
			$U_j\gets\max\{4\Delta,8(n/w)d_W(\sigma_j)\}$
			\gray{\tcc*{a constant-factor degree upper bound}}
			
			$a_j\gets U_j^p$, $B_H\gets B_H+a_j$
			\gray{\tcc*{update the total high-degree weight}}
			
			\For{$r\gets1,\ldots,t$}{
				with probability $a_j/B_H$, set $H[r]\gets(\sigma_j,U_j)$
				\gray{\tcc*{weighted sampling}}
			}
		}
	}
	\Return $(W,L,D_L,H,B_H,P_1,\ldots,P_K)$
\end{algorithm}

\begin{algorithm}[!ht]
	\caption{Output by repeated rejection sampling}
	\label{alg:output}
	\DontPrintSemicolon
	\SetAlgoNoEnd
	
	\KwIn{the state returned by Algorithm~\ref{alg:preprocess}}
	\KwOut{an $L_{M_p}$ sample, or $\FAIL$}
	
	$B_L\gets n\Delta^p$, $B\gets B_L+B_H$, $s\gets1$
	\gray{\tcc*{low- and high-degree normalizing weights}}
	
	\For{$r\gets1,\ldots,t$}{
		choose $Z_r\in\{\Low,\High\}$ with $\Pr[Z_r=\Low]=B_L/B$
		
		\eIf{$Z_r=\Low$}{
			\If{$d_W(L[r])<\theta$ and $D_L[r]\le\Delta$}{
				with probability $(D_L[r]/\Delta)^p$, \Return $L[r]$
				\gray{\tcc*{use the exact low degree}}
			}
		}{
			\If{$H[r]\ne\FAIL$}{
				let $H[r]=(v,U)$
				\gray{\tcc*{$v$ was sampled with weight $U^p$}}
				
				start $\ExactReject$ as in Lemma~\ref{lem:exact-rejection}
				
				\While{$\ExactReject$ requests another input bit}{
					\If{$s>K$}{
						\Return $\FAIL$
						\gray{\tcc*{the degree-witness sets are exhausted}}
					}
					give it $\DegreeCoin(v,U;s)$ and set $s\gets s+1$
					\gray{\tcc*{the bit has probability $1-e^{-d(v)/U}$}}
				}
				\lIf{$\ExactReject$ outputs $1$}{
					\Return $v$
				}
			}
		}
	}
	\Return $\FAIL$
\end{algorithm}

\vspace{2mm}
\noindent{\bf The analysis.\ }
For the analysis, define the \emph{idealized output procedure} to be
Algorithm~\ref{alg:output} with an unlimited supply of independent
degree-witness sets.

Let
$\cH\coloneqq\{j\in[n]:d_W(\sigma_j)\ge\theta\}$ and
$\cL\coloneqq[n]\setminus\cH.$

Define $\cE_W$ to be the event that all of the following properties hold: (1) $d_j<\Delta$ for every $j\in\cL$; (2) $d_j>\Delta/4$ for every $j\in\cH$; and (3) $\wt d_j\in[d_j/2,2d_j]$ for every $j\in\cH$.

\begin{lemma}\label{lem:witness}
	$\Pr[\cE_W]\ge1-n^{-(c+8)}.$
	Conditioned on $\cE_W$, every $j\in\cH$ satisfies
	$4d_j\le U_j\le16d_j,$
	$
	\frac1{16}\le\frac{d_j}{U_j}\le\frac14,
	$
	and $U_j\ge4\Delta.$
\end{lemma}

\begin{proof}
	Fix $j\in[n]$ and let $X_j=d_W(\sigma_j)$.  Since the entries of $W$ are
	independent uniform samples from the stream, the random variables
	$\{sim(W[a],\sigma_j) : a\in[w]\}$ are independent and lie in
	$[0,1]$.  We have
	$
	X_j=\sum_{a\in[w]}sim(W[a],\sigma_j)
	$ 
	and
	$
	\mu_j\coloneqq\E[X_j] = C_W\frac{d_j}{\Delta}\ln n.
	$
	
	If $d_j\ge\Delta$, then $\mu_j\ge C_W\ln n$ and
	$\theta\le\mu_j/2$.  A Chernoff bound gives
	\begin{equation}
		\label{eq:witness-high}
		\Pr[X_j<\theta]
		\le\Pr[X_j<\mu_j/2]
		\le e^{-\mu_j/8}
		\le n^{-C_W/8}.
	\end{equation}
	
	If $d_j\le\Delta/4$, then $\mu_j\le\mu_0=(C_W/4)\ln n$ and
	$\theta=2\mu_0$.  A Chernoff bound gives
	\begin{equation}
		\Pr[X_j\ge\theta]
		\le e^{-\mu_0/3}
		\le n^{-C_W/12}.
		\label{eq:witness-low}
	\end{equation}
	
	Finally, suppose that $d_j>\Delta/4$.  Then
	$\mu_j>(C_W/4)\ln n$, and a Chernoff bound gives
	\begin{equation}
		\Pr\left[X_j\notin[\mu_j/2,2\mu_j]\right]
		\le e^{-\mu_j/8}+e^{-\mu_j/3}
		\le2n^{-C_W/32}.
		\label{eq:witness-estimate}
	\end{equation}
	A union bound over all nodes in \eqref{eq:witness-high}--\eqref{eq:witness-estimate}
	shows that the probability that any of the three defining properties of
	$\cE_W$ fails is at most
	$$
	n^{1-C_W/8}+n^{1-C_W/12}+2n^{1-C_W/32}  \le n^{-(c+8)},
	$$
	where the last inequality uses $C_W=64(c+10)$.
	
	For $j\in\cH$, the event $\cE_W$ gives $8\wt d_j\in[4d_j,16d_j]$ and $d_j>\Delta/4$.  Hence, $4\Delta<16d_j$.  Taking the maximum in
	\eqref{eq:degree-upper-bound} preserves the upper bound $16d_j$ and gives the lower bounds $U_j\ge4d_j$ and $U_j\ge4\Delta$.  Dividing by $U_j$ gives the remaining inequalities.
\end{proof}

The following lemma shows that the degree-witness sets use bounded space while providing the exact degree-dependent rejection probability for each high-degree sample.

\begin{lemma}
	\label{lem:storage-correction}
	The probability that Pass~1 stores more than $S_{\max}=2Kn/\Delta$ entries in the degree-witness sets is at most $n^{-(c+8)}$.  Conditioned on $\cE_W$, if $(v,U)$ is any high-degree sample, then $\ExactReject$ (see Lemma~\ref{lem:exact-rejection}), supplied with
	independent bits $\DegreeCoin(v,U;s)$, accepts $v$ with probability $\left(\frac{d(v)}{U}\right)^p.$
\end{lemma}

\begin{proof}
	Each of the $nK$ (node, set) pairs is stored independently with probability
	$$
	\tau_0=1-e^{-1/(4\Delta)}\le \frac{1}{4\Delta}.
	$$
	Therefore, the total number of stored entries is stochastically dominated by
	a binomial random variable with mean
	$
	\mu=\frac{Kn}{4\Delta}.
	$
	Since $S_{\max}=8\mu$, a Chernoff bound gives
	$$
	\Pr\bigg[\sum_{s \in [K]}|P_s|>S_{\max}\bigg] \le\exp\left(-\frac{7Kn}{12\Delta}\right) \le\exp\left(-\frac{7K}{12}\right) \le n^{-(c+8)},
	$$
	where the last inequality follows from
	$K=C_K\ln n$ and $C_K\ge12(c+8)/7$.
	
	Now condition on $\cE_W$ and fix a high-degree sample $(v,U)$.  By
	Lemma~\ref{lem:witness}, $d(v)/U\in[1/16,1/4]$.  Equation
	\eqref{eq:degree-coin-probability} shows that every input bit has probability
	$$
	q=1-e^{-d(v)/U} \in\left[1-e^{-1/16},1-e^{-1/4}\right].
	$$
	Lemma~\ref{lem:exact-rejection} therefore gives
	$$
	\Pr[\ExactReject\text{ outputs }1] =\phi_p(q) = \bigl(-\ln(e^{-d(v)/U})\bigr)^p = \left(\frac{d(v)}U\right)^p.
	$$
\end{proof}

We next bound the number of degree-witness sets needed by the algorithm.
\begin{lemma}\label{lem:total-coins}
	Conditioned on $\cE_W$, the probability that the idealized output procedure
	uses more than $K$ degree-witness sets is at most $n^{-(c+8)}$.
\end{lemma}

\begin{proof}
	We first derive two bounds for one call to $\ExactReject$.  Let $N$ be the
	number of input bits used and let $A$ be its output.  For every $q\in \left[1-e^{-1/16},1-e^{-1/4}\right]$, using $e^{\eta_p N}-1 = (e^{\eta_p}-1)\sum_{m=0}^{N-1}e^{\eta_p m}$ and taking expectations, we obtain
	\begin{equation}
		\E[e^{\eta_p N}]-1 = (e^{\eta_p}-1) \sum_{m\ge0}e^{\eta_p m}\Pr[N>m] \le \frac{a_p(e^{\eta_p}-1)}{1-e^{\eta_p}\rho_p}
		\le\frac{\alpha_p}{2},
		\label{eq:one-call-moment}
	\end{equation}
	where the last inequality is due to \eqref{eq:lambda-choice}.
	Also, Lemma~\ref{lem:storage-correction} gives
	$\Pr[A=1]\ge\alpha_p$.  It follows that
	\begin{equation}
		\E[e^{\eta_p N}\one[A=0]] \le\Pr[A=0]+\E[e^{\eta_p N}-1] \le1-\frac{\alpha_p}{2},
		\label{eq:reject-moment}
	\end{equation}
	and
	\begin{equation}
		\E[e^{\eta_p N}\one[A=1]] \le \E[e^{\eta_p N}] \stackrel{\text{by}\ \eqref{eq:one-call-moment}}{\le} 1+\frac{\alpha_p}{2}.
		\label{eq:accept-moment}
	\end{equation}
	
	Fix $W$ satisfying $\cE_W$, and condition on the random choices defining $L[1..t]$, $H[1..t]$, and $Z_1,\ldots,Z_t$.  Ignoring possible acceptance by the low-degree procedure can only increase the number of degree bits used, so the remaining high-degree calls form a fixed sequence.  For the analysis, assign each such call an independent infinite supply of degree-witness sets and random bits; this is distributionally equivalent to using the next unused degree-witness set whenever a bit is requested.
	
	Let $(N_i,A_i)$ denote the number of input bits and the output of the $i$-th high-degree call.  For the analysis, consider the sequence of high-degree calls made by the output procedure. If none accepts, append independent hypothetical high-degree calls, each satisfying the same bounds as an actual call, until the first acceptance occurs. Let $G$ denote the index of this first accepting call.  Since every call accepts with probability at least $\alpha_p$,
	$
	\Pr[G>k]\le(1-\alpha_p)^k,
	$
	so $G$ is finite with probability one.
	
	Define
	$S=\sum_{i \in [G]} N_i$.  Conditional on the fixed sequence, the pairs
	$(N_i,A_i)$ are independent.  Equations \eqref{eq:reject-moment} and
	\eqref{eq:accept-moment} therefore imply
	\begin{eqnarray}
		\E[e^{\eta_p S}]
		&=&\sum_{k\ge1} \E\left[e^{\eta_p\sum_{i=1}^kN_i}\one[A_1=0,\ldots,A_{k-1}=0,A_k=1]\right]\nonumber\\
		&\le& \sum_{k\ge1}\left(1-\frac{\alpha_p}{2}\right)^{k-1}
		\left(1+\frac{\alpha_p}{2}\right) = \frac{1+\alpha_p/2}{\alpha_p/2}.
		\label{eq:total-moment}
	\end{eqnarray}
	Since \eqref{eq:total-moment} holds for every $W$ satisfying $\cE_W$
	and every fixing of the conditioned random choices, averaging over these
	choices gives
	\begin{equation}
		\E[e^{\eta_p S}\mid\cE_W]
		\le
		\frac{1+\alpha_p/2}{\alpha_p/2}.
		\label{eq:total-moment-condition}
	\end{equation}
	
	Let $T$ be the total number of
	$\DegreeCoin$ bits used before the procedure terminates. We always have $T \le S$, since the output procedure may stop at a low-degree sample or after its $t$ trials.  Markov's inequality and \eqref{eq:total-moment-condition} give
	\begin{eqnarray*}
		\Pr[T>m\mid\cE_W] \le \Pr[S>m\mid\cE_W]
		&=& \Pr[e^{\eta_p S}>e^{\eta_p m}\mid\cE_W]\\
		&\le& e^{-\eta_p m}\E[e^{\eta_p S}\mid\cE_W]\\
		&\le& \frac{1+\alpha_p/2}{\alpha_p/2}e^{-\eta_p m}.
	\end{eqnarray*}
	
	Finally, since $K=C_K\ln n$ and $C_K\ge(c+9)/\eta_p$, we have
	$
	\frac{1+\alpha_p/2}{\alpha_p/2} e^{-\eta_p K} \le n^{-(c+8)}
	$.  
\end{proof}

Fix $W$ satisfying $\cE_W$, and define
$B_L\coloneqq n\Delta^p$, $B_H\coloneqq\sum_{j\in\cH}U_j^p$,
and $B\coloneqq B_L+B_H.$

The following lemma shows that each trial produces exactly the desired degree-proportional weights and that the repeated trials succeed with high probability.

\begin{lemma}\label{lem:sampling}
	Conditioned on $W$ and $\cE_W$, each trial of the idealized output procedure returns $\sigma_j$ with probability $d_j^p/B$ for every $j\in[n]$. Hence, conditioned on at least one trial succeeding, the first returned node is distributed as $(d_j^p/M_p)_{j\in[n]}$. Moreover,
	$
	\Pr[\text{all $t$ trials fail}\mid W,\cE_W]\le n^{-2(c+10)}.
	$
\end{lemma}

\begin{proof}
	The update rule for $H[r]$ is weighted reservoir sampling with weights $U_j^p$.  Therefore, when $\cH\ne\emptyset$,
	\begin{equation}
		\Pr[H[r]=(\sigma_j,U_j)\mid W] = \frac{U_j^p}{B_H}
		\qquad \forall j\in\cH.
		\label{eq:weighted-sample}
	\end{equation}
	
	Now fix $j\in\cL$.  The trial chooses the low-degree procedure with
	probability $B_L/B$, and $L[r]$ equals $\sigma_j$ with probability $1/n$.
	Conditioned on $\cE_W$, the two checks in Algorithm~\ref{alg:output} pass and
	$D_L[r]=d_j<\Delta$.  Therefore,
	\begin{equation}
		\Pr[\text{one trial returns }\sigma_j]
		= \frac{B_L}{B}\cdot\frac1n\cdot \left(\frac{d_j}{\Delta}\right)^p = \frac{d_j^p}{B}.
		\label{eq:low-one-trial}
	\end{equation}
	
	For a fixed $j\in\cH$, by \eqref{eq:weighted-sample} and Lemma~\ref{lem:storage-correction},
	\begin{equation}
		\Pr[\text{one trial returns }\sigma_j]
		= \frac{B_H}{B}\cdot\frac{U_j^p}{B_H}\cdot
		\left(\frac{d_j}{U_j}\right)^p = \frac{d_j^p}{B}.
		\label{eq:high-one-trial}
	\end{equation}
	Equations \eqref{eq:low-one-trial} and \eqref{eq:high-one-trial} prove the first statement of the lemma.  
	Summing over all nodes shows that one trial succeeds
	with probability $q\coloneqq\frac{M_p}{B}.$
	The trials are independent; as in the proof of Lemma~\ref{lem:total-coins}, assigning a separate independent sequence of degree-witness sets to each trial does not change their distribution. Hence,
	$$
	\Pr[\text{the first returned node is }\sigma_j] = \frac{d_j^p}{B}\sum_{r=0}^{t-1}(1-q)^r.
	$$
	Conditioning on the event that at least one trial succeeds cancels the common
	geometric factor and gives $d_j^p/M_p$.
	
	It remains to analyze the probability that all $t$ trials fail.  Lemma~\ref{lem:witness}
	gives
	$$
	B_H=\sum_{j\in\cH}U_j^p \le 16^p\sum_{j\in\cH}d_j^p \le16^pM_p.
	$$
	Since $M_p\ge n$,
	$
	B_L=n\Delta^p\le\Delta^pM_p.
	$
	Therefore,
	$$
	q=\frac{M_p}{B} \ge \frac1{\Delta^p+16^p}.
	$$
	Using the definition of $t$,
	$$
	\Pr[\text{all $t$ trials fail}\mid W,\cE_W] = (1-q)^t \le \exp\left(-\frac{t}{\Delta^p+16^p}\right) = n^{-2(c+10)}.
	$$
\end{proof}

\begin{theorem}\label{thm:main}
	For every symmetric similarity function
	$sim:\cU\times\cU\to[0,1]$ satisfying
	$sim(x,x)=1$ for every $x\in\cU$, and every fixed $p>0$ and $c>0$, Algorithms~\ref{alg:preprocess}
	and~\ref{alg:output} form an $n^{-c}$-perfect $L_{M_p}$-sampler.  The algorithm uses two passes and $O_{p,c}\left(n^{1-1/(p+1)}\log n\right)$
	words of space.
\end{theorem}

\begin{proof}
	Let $\cB$ be the event that at least one of the following occurs: $\cE_W$ fails, the storage cap $S_{\max}$ is exceeded, or the idealized output procedure uses more than $K$ degree-witness sets.  By Lemmas~\ref{lem:witness}--\ref{lem:total-coins},
	$\Pr[\cB]\le 3n^{-(c+8)}.$
	If $\cB$ does not happen, the algorithm behaves exactly as the idealized output procedure.  By Lemma~\ref{lem:sampling}, the latter returns a node according to the target distribution $(d_j^p/M_p)_{j\in[n]}$, conditioned on returning
	a node, and all $t$ trials fail with probability at most
	$n^{-2(c+10)}$.  Hence, the algorithm returns $\FAIL$ with probability at most
	$
	\Pr[\cB]+n^{-2(c+10)} \le 4n^{-(c+8)} \le n^{-c}.
	$
	
	Let $X$ and $\wt X$ denote the outputs of the algorithm and of the idealized procedure, and let $\pi=(d_j^p/M_p)_{j\in[n]}$.  By Lemma~\ref{lem:sampling}, $\Pr[\wt X\in A\mid W]$ is a $W$-dependent scalar times $\pi(A)$; averaging over $W$ satisfying $\cE_W$ gives $\Pr[\wt X\in A,\cE_W]=r\cdot\pi(A)$ for $r\coloneqq\Pr[\wt X\ne\FAIL,\,\cE_W]$.  The two executions agree off $\cB$, so $|\Pr[X\in A]-r\pi(A)|\le\Pr[\cB]$ and
	$|\Pr[X\ne\FAIL]-r|\le\Pr[\cB]$. Therefore,
	$$
	\left|\Pr[X\in A\mid X\ne\FAIL]-\pi(A)\right| \le\frac{2\Pr[\cB]}{\Pr[X\ne\FAIL]}\le n^{-c}.
	$$
	Taking the supremum over $A$ proves the sampling guarantee.
	
	It remains to bound the space usage.  The witness set uses $w=O_{p,c}\left(\frac n\Delta\log n\right)$
	words.  The arrays $L$, $D_L$, and $H$ use
	$O_{p,c}(t)=O_{p,c}(\Delta^p\log n)$ words.  The degree-witness sets use at most
	$S_{\max}=O_{p,c}\left(\frac n\Delta\log n\right)$
	words.  Since $\Delta^p=n/\Delta=n^{p/(p+1)}$, the total space is
	$O_{p,c}\left(n^{p/(p+1)}\log n\right) = O_{p,c}\left(n^{1-1/(p+1)}\log n\right).$
	The pass complexity follows directly from the algorithm description.
\end{proof}

\subsection{Lower Bounds for Perfect $L_{M_p}$-Sampling}
\label{sec:LMp-lower-bounds}

We reduce degree-moment estimation to perfect $L_{M_p}$-sampling. 

Fix $\eps\in(0,1/20]$, and let
$\ell_{\min}\coloneqq\lfloor\log_2 n\rfloor$,
$\ell_{\max}\coloneqq\lceil(p+1)\log_2 n\rceil$, and
$\mathcal I\coloneqq\{\ell_{\min},\ldots,\ell_{\max}\}$.
For every $\ell\in\mathcal I$, set
$x_\ell\coloneqq\lceil2^{\ell/(p+1)}\rceil$ and
$W_\ell\coloneqq x_\ell^{p+1}$.  Clearly, for $n$ sufficiently large,
\begin{equation}
	2^\ell\le W_\ell<2^{\ell+1},
	\qquad \forall \ell\in\mathcal I.
	\label{eq:dummy-scale}
\end{equation}

Let $\sigma^{(\ell)}$ be obtained from $\sigma$ by appending $x_\ell$ new nodes that form a clique and are dissimilar to every original node.  The original degrees are unchanged, while every dummy node has degree $x_\ell$. Consequently, $M_p(\sigma^{(\ell)})=M_p(\sigma)+W_\ell$.
Moreover, $x_\ell\le2n+1$ for every $\ell\in\mathcal I$, so every augmented stream has at most $4n$ nodes.

\begin{algorithm}[t]
	\caption{$M_p$-Estimation from Perfect $L_{M_p}$-Sampling}
	\label{alg:Mp-from-LMp}
	\DontPrintSemicolon
	\SetAlgoNoEnd
	
	\KwIn{a stream $\sigma$ of $n$ nodes, an $n^{-c}$-perfect
		$L_{M_p}$-sampler $\mathcal A$, and $\eps\in(0,1/20]$}
	\KwOut{a $(1+\eps)$-approximation to $M_p(\sigma)$, or $\FAIL$}
	
	$b\gets2^{16}\eps^{-2}\ln n$, $x_\ell\coloneqq\lceil2^{\ell/(p+1)}\rceil$,
	$W_\ell\coloneqq x_\ell^{p+1}$
	
	$\mathcal I \gets \{\lfloor\log_2 n\rfloor,\ldots,\lceil(p+1)\log_2 n\rceil\}$.
	
	\ForEach{$\ell\in\mathcal I$}{
		form $\sigma^{(\ell)}$ by appending a dummy clique of size $x_\ell$
		
		run $b$ independent copies of $\mathcal A$ on $\sigma^{(\ell)}$
		
		$Y^{(\ell)}\gets$ number of non-failure outputs
		
		$X^{(\ell)}\gets$ number of dummy-node outputs
		
		\If{$Y^{(\ell)}\ge b/3$ and
			$X^{(\ell)}/Y^{(\ell)}\in[1/4,3/4]$}{
			\Return
			$\displaystyle
			\left({Y^{(\ell)}}/{X^{(\ell)}}-1\right)W_\ell$
		}
	}
	\Return $\FAIL$
\end{algorithm}

All copies in Algorithm~\ref{alg:Mp-from-LMp}, for all values of $\ell$, are run in parallel.  In each pass, every copy first receives the original stream $\sigma$ and then receives its own dummy clique.  Thus, if $\mathcal A$ uses $\kappa$ passes, the reduction also uses $\kappa$ passes.

\begin{lemma}
	\label{lem:LMp-to-Mp}
	Fix constants $p>0$, $c>0$, and $\eps\in(0,1/20]$.  Suppose $\mathcal A$ is a $\kappa$-pass, $s$-space, $n^{-c}$-perfect $L_{M_p}$-sampler on streams of length $O(n)$.  Then Algorithm~\ref{alg:Mp-from-LMp} is a $\kappa$-pass
	$(1+\eps)$-approximation algorithm for $M_p$ with failure probability $o(1)$ and space $O_p(s\eps^{-2}\log^2 n)$.
\end{lemma}

\begin{proof}
	Fix $\ell\in\mathcal I$ and write
	$q_\ell=W_\ell/(W_\ell+M_p)$.  Since the total $p$-th degree mass of the dummy clique is $W_\ell$, the quantity $q_\ell$ is the probability that an exact $L_{M_p}$ sample from $\sigma^{(\ell)}$ belongs to the dummy clique. Let $p_\ell$ be the corresponding probability conditioned on $\mathcal A$
	returning a node.  Since $\mathcal A$ is $n^{-c}$-perfect and $\sigma^{(\ell)}$ has $O(n)$ nodes, we have $|p_\ell-q_\ell|\le\gamma_\ell$, where
	$\gamma_\ell=O(n^{-c})$.  Thus,
	$\gamma_\ell\le\eps/64$ simultaneously for every $\ell$.
	
	Let $s_\ell$ be the success probability of one copy of $\mathcal A$ on $\sigma^{(\ell)}$.  We have $s_\ell\ge0.51$. Moreover,
	$$
	Y^{(\ell)}\sim\operatorname{Binomial}(b,s_\ell)
	\quad\text{and}\quad \left(X^{(\ell)}\mid Y^{(\ell)}=y\right)\sim\operatorname{Binomial}(y,p_\ell).
	$$
	By Hoeffding's inequality and the choice of $b$, for each fixed $\ell$,
	\begin{equation}
		\Pr\Bigg[\left(Y^{(\ell)}<\frac b3\right)
		\vee \left(\left| \frac{X^{(\ell)}}{Y^{(\ell)}}-p_\ell
		\right|>\frac{\eps}{64}\right)\Bigg] \le n^{-10}
		+ 2\exp\left(-\frac{b\eps^2}{6144}\right)
		\le 3n^{-10},
		\label{eq:LMp-one-scale-concentration}
	\end{equation}
	where the ratio ${X^{(\ell)}}/{Y^{(\ell)}}$ is considered only when $Y^{(\ell)}\ge b/3$.  A union bound
	over the $O_p(\log n)$ choices of $\ell$ shows that, with probability at least $1-n^{-9}$, simultaneously for every $\ell\in\mathcal I$,
	\begin{equation}
		Y^{(\ell)}\ge\frac b3
		\quad\text{and}\quad
		\left|\frac{X^{(\ell)}}{Y^{(\ell)}}-q_\ell\right|\le\frac{\eps}{32}.
		\label{eq:all-scales-accurate}
	\end{equation}
	We condition on this event.
	
	Suppose the algorithm returns at scale $\ell$, and let
	$
	\widehat q_\ell\coloneqq X^{(\ell)}/Y^{(\ell)}.
	$
	The returned estimate $\widehat M_\ell$ and the true moment $M_p$ satisfy
	$
	\widehat M_\ell
	=
	W_\ell\frac{1-\widehat q_\ell}{\widehat q_\ell}
	$
	and
	$
	M_p
	=
	W_\ell\frac{1-q_\ell}{q_\ell}.
	$
	Therefore,
	\begin{equation}
		\left|\frac{\widehat M_\ell}{M_p}-1\right|
		=
		\left|
		\frac{q_\ell(1-\widehat q_\ell)}
		{\widehat q_\ell(1-q_\ell)}
		-1
		\right| 
		=
		\frac{|q_\ell-\widehat q_\ell|}
		{\widehat q_\ell(1-q_\ell)}.
		\label{eq:Mp-relative-error}
	\end{equation}
	Because scale $\ell$ is accepted, $\widehat q_\ell\in[\frac14, \frac34]$.  By
	\eqref{eq:all-scales-accurate} and $\eps\le1/20$, we have
	$q_\ell\le\frac34+\frac{\eps}{32}\le\frac{481}{640}$ and 
	$
	\widehat q_\ell(1-q_\ell)
	\ge\frac14\cdot\frac{159}{640}
	=\frac{159}{2560}.
	$
	Substituting these bounds into \eqref{eq:Mp-relative-error} gives
	$|\widehat M_\ell/M_p-1|\le(80/159)\eps<\eps$.  Hence, every returned value
	is a $(1+\eps)$-approximation to $M_p$.
	
	It remains to show that the algorithm returns a value.  Let
	$\ell^*=\lfloor\log M_p\rfloor$.  Since $n\le M_p\le n^{p+1}$, we have
	$\ell^*\in\mathcal I$.  Equation~\eqref{eq:dummy-scale} implies
	$$
	2^{\ell^*}\le M_p,W_{\ell^*}<2^{\ell^*+1}, \quad \text{and} \quad \frac13<q_{\ell^*}<\frac23.
	$$
	By \eqref{eq:all-scales-accurate},
	$X^{(\ell^*)}/Y^{(\ell^*)}\in[1/4,3/4]$.  Thus scale $\ell^*$ is accepted
	unless an earlier scale has already been accepted, and in either case the
	algorithm returns a $(1+\eps)$-approximation to $M_p$.
	
	Finally, the algorithm considers $O_p(\log n)$ scales and runs
	$b=O(\eps^{-2}\log n)$ copies of $\mathcal A$ at each scale.  All copies are
	run in parallel, so the number of passes remains $\kappa$ and the total space
	usage is $O_p(s\eps^{-2}\log^2 n)$.
\end{proof}

We use the following one-pass lower bound for degree-moment estimation in the previous
work.

\begin{theorem}[\cite{LZ26b}]
	\label{thm:Mp-estimation-lower-bound}
	Fix a constant $p>0$.  For every constant $C\ge1$, any one-pass algorithm that returns a $C$-approximation to $M_p$ with probability at least $0.51$ requires $\Omega(n)$ bits of space.
\end{theorem}

\begin{theorem}
	\label{thm:LMp-one-pass-lb}
	Fix constants $p>0$ and $c>0$.  Any one-pass $n^{-c}$-perfect
	$L_{M_p}$-sampler requires
	$\Omega_p(n/\log^2 n)$ bits of space.
\end{theorem}

\begin{proof}
	Suppose the sampler uses $s$ bits on streams of length at most $n$, and set $m= n/4$.  Apply Lemma~\ref{lem:LMp-to-Mp} with $\eps=0.005$ to an $m$-node input.  Every augmented stream has length at most $4m = n$, so the resulting one-pass estimator uses $O_p(s\log^2 n)$ bits and succeeds with probability greater than $0.51$.  Its output is  a $1.005$-approximation to $M_p$.
	Theorem~\ref{thm:Mp-estimation-lower-bound} therefore gives	$s\log^2 n=\Omega(m)=\Omega(n)$, proving the theorem.
\end{proof}

For a constant number of passes, we use the following lower bound from previous work.

\begin{theorem}[\cite{LZ26b}]
	\label{thm:Mp-constant-pass-lb}
	Fix a constant $p>0$ and let $\eps\in(0,0.9)$.  Any constant-pass $(1+\eps)$-approximation algorithm for $M_p$ with probability at least $0.51$ requires
	$
	\Omega_p\left(\min\left\{n,\eps^{-\frac{1}{p+1}}n^{1-\frac{1}{p+1}}\right\}\right)
	$
	bits of space.
\end{theorem}

\begin{theorem}
	\label{thm:LMp-constant-pass-lb}
	Fix constants $p>0$, $c>0$, and an integer $\kappa\ge1$.  Any $\kappa$-pass
	$n^{-c}$-perfect $L_{M_p}$-sampler requires
	$
	\Omega_p\left({n^{1-\frac{1}{p+1}}}/{\log^2 n}\right)
	$
	bits of space.
\end{theorem}

\begin{proof}
	Suppose the sampler uses $s$ bits on streams of length at most $n$, and set
	$m= n/4$.  Apply Lemma~\ref{lem:LMp-to-Mp} with
	$\eps=0.005$ to an $m$-node input.  The resulting $\kappa$-pass estimator uses $O_p(s\log^2 n)$ bits, succeeds with probability greater than $0.51$,
	and computes a $(1.005)$-approximation to $M_p$.  Since $\kappa$ and $\eps$ are fixed, Theorem~\ref{thm:Mp-constant-pass-lb} gives
	$
	s\log^2 n = \Omega_p\left(m^{p/(p+1)}\right) = \Omega_p\left(n^{p/(p+1)}\right).
	$
	Rearranging proves the theorem.
\end{proof}

\section*{Acknowledgements}
The author thanks Kaiwen Liu for helpful discussions regarding the lower bound proof and the $L_{\DI}$-sampling algorithm in the early stage of this work.
	
\section*{Use of AI assistance}
The research questions, algorithms, and mathematical arguments in this paper were developed by the author. ChatGPT (GPT-5.6) was used for language editing, proofreading, and literature search (in particular, for pointing us to Nacu and Peres~\cite{NP05}, which enabled a modular treatment of one component of the $L_{M_p}$-sampler). The author takes full responsibility for the paper.
	

\begin{thebibliography}{BJK{\etalchar{+}}02}
	
	\bibitem[AKO11]{AKO11}
	Alexandr Andoni, Robert Krauthgamer, and Krzysztof Onak.
	\newblock Streaming algorithms via precision sampling.
	\newblock In Rafail Ostrovsky, editor, {\em FOCS}, pages 363--372. {IEEE}
	Computer Society, 2011.
	
	\bibitem[AMS99]{AMS99}
	Noga Alon, Yossi Matias, and Mario Szegedy.
	\newblock The space complexity of approximating the frequency moments.
	\newblock {\em J. Comput. Syst. Sci.}, 58(1):137--147, 1999.
	
	\bibitem[BJK{\etalchar{+}}02]{BJKST02}
	Ziv Bar{-}Yossef, T.~S. Jayram, Ravi Kumar, D.~Sivakumar, and Luca Trevisan.
	\newblock Counting distinct elements in a data stream.
	\newblock In {\em RANDOM}, pages 1--10, 2002.
	
	\bibitem[CDK18]{CDK18}
	Graham Cormode, Jacques Dark, and Christian Konrad.
	\newblock Approximating the caro-wei bound for independent sets in graph
	streams.
	\newblock In {\em ISCO}, pages 101--114, 2018.
	
	\bibitem[CMR05]{CMR05}
	Graham Cormode, S.~Muthukrishnan, and Irina Rozenbaum.
	\newblock Summarizing and mining inverse distributions on data streams via
	dynamic inverse sampling.
	\newblock In {\em VLDB}, pages 25--36, 2005.
	
	\bibitem[cNP05]{NP05}
	\c{S}erban Nacu and Yuval Peres.
	\newblock Fast simulation of new coins from old.
	\newblock {\em The Annals of Applied Probability}, 15(1A):93--115, 2005.
	
	\bibitem[CP17]{CP17}
	Sergio Cabello and Pablo P{\'{e}}rez{-}Lantero.
	\newblock Interval selection in the streaming model.
	\newblock {\em Theor. Comput. Sci.}, 702:77--96, 2017.
	
	\bibitem[CZ16]{CZ16}
	Di~Chen and Qin Zhang.
	\newblock Streaming algorithms for robust distinct elements.
	\newblock In {\em SIGMOD}, pages 1433--1447. {ACM}, 2016.
	
	\bibitem[CZ18]{CZ18}
	Jiecao Chen and Qin Zhang.
	\newblock Distinct sampling on streaming data with near-duplicates.
	\newblock In {\em PODS}, pages 369--382, 2018.
	
	\bibitem[EHR16]{EHR16}
	Yuval Emek, Magn{\'{u}}s~M. Halld{\'{o}}rsson, and Adi Ros{\'{e}}n.
	\newblock Space-constrained interval selection.
	\newblock {\em {ACM} Trans. Algorithms}, 12(4):51:1--51:32, 2016.
	
	\bibitem[ERS17]{ERS17}
	Talya Eden, Dana Ron, and C.~Seshadhri.
	\newblock Sublinear time estimation of degree distribution moments: The
	degeneracy connection.
	\newblock In {\em ICALP}, pages 7:1--7:13, 2017.
	
	\bibitem[Fei04]{Feige04}
	Uriel Feige.
	\newblock On sums of independent random variables with unbounded variance, and
	estimating the average degree in a graph.
	\newblock In {\em STOC}, pages 594--603, 2004.
	
	\bibitem[FIS05]{FIS05}
	Gereon Frahling, Piotr Indyk, and Christian Sohler.
	\newblock Sampling in dynamic data streams and applications.
	\newblock In {\em SoCG}, pages 142--149, 2005.
	
	\bibitem[Gib01]{Gibbons01}
	Phillip~B. Gibbons.
	\newblock Distinct sampling for highly-accurate answers to distinct values
	queries and event reports.
	\newblock In {\em VLDB}, pages 541--550, 2001.
	
	\bibitem[GR08]{GR08}
	Oded Goldreich and Dana Ron.
	\newblock Approximating average parameters of graphs.
	\newblock {\em Random Struct. Algorithms}, 32(4):473--493, 2008.
	
	\bibitem[GRS11]{GRS11}
	Mira Gonen, Dana Ron, and Yuval Shavitt.
	\newblock Counting stars and other small subgraphs in sublinear-time.
	\newblock {\em {SIAM} J. Discret. Math.}, 25(3):1365--1411, 2011.
	
	\bibitem[Ind06]{Ind06}
	Piotr Indyk.
	\newblock Stable distributions, pseudorandom generators, embeddings, and data
	stream computation.
	\newblock {\em J. {ACM}}, 53(3):307--323, 2006.
	
	\bibitem[IW05]{IW05}
	Piotr Indyk and David~P. Woodruff.
	\newblock Optimal approximations of the frequency moments of data streams.
	\newblock In {\em STOC}, pages 202--208, 2005.
	
	\bibitem[JST11]{JST11}
	Hossein Jowhari, Mert Saglam, and G{\'{a}}bor Tardos.
	\newblock Tight bounds for lp samplers, finding duplicates in streams, and
	related problems.
	\newblock In {\em PODS}, pages 49--58. {ACM}, 2011.
	
	\bibitem[JW18]{JW18}
	Rajesh Jayaram and David~P. Woodruff.
	\newblock Perfect lp sampling in a data stream.
	\newblock In {\em FOCS}, pages 544--555, 2018.
	
	\bibitem[JWZ22]{JWZ22}
	Rajesh Jayaram, David~P. Woodruff, and Samson Zhou.
	\newblock Truly perfect samplers for data streams and sliding windows.
	\newblock In {\em PODS}, pages 29--40. {ACM}, 2022.
	
	\bibitem[KNW10]{KNW10}
	Daniel~M. Kane, Jelani Nelson, and David~P. Woodruff.
	\newblock An optimal algorithm for the distinct elements problem.
	\newblock In {\em PODS}, pages 41--52, 2010.
	
	\bibitem[LVZ26]{LVZ26}
	Kaiwen Liu, Seba~Daniela Villalobos, and Qin Zhang.
	\newblock Estimating correlation clustering cost in node-arrival stream.
	\newblock {\em CoRR}, abs/2605.07091, 2026.
	
	\bibitem[LZ26a]{LZ26Mismatch}
	Kaiwen Liu and Qin Zhang.
	\newblock Frequency moments in noisy streaming and distributed data under
	mismatch ambiguity.
	\newblock {\em Proc. {ACM} Manag. Data}, 4(2):105:1--105:24, 2026.
	
	\bibitem[LZ26b]{LZ26b}
	Kaiwen Liu and Qin Zhang.
	\newblock Statistics of similarity graphs in node-arrival streams, 2026.
	\newblock arXiv:2609.04505.
	
	\bibitem[McG14]{McGregor14}
	Andrew McGregor.
	\newblock Graph stream algorithms: a survey.
	\newblock {\em {SIGMOD} Rec.}, 43(1):9--20, 2014.
	
	\bibitem[MW10]{MW10}
	Morteza Monemizadeh and David~P. Woodruff.
	\newblock 1-pass relative-error l\({}_{\mbox{p}}\)-sampling with applications.
	\newblock In {\em SODA}, pages 1143--1160, 2010.
	
	\bibitem[Nis90]{Nisan90}
	Noam Nisan.
	\newblock Pseudorandom generators for space-bounded computation.
	\newblock In {\em STOC}, pages 204--212, 1990.
	
	\bibitem[PRV01]{PRV01}
	Stephen Ponzio, Jaikumar Radhakrishnan, and Srinivasan Venkatesh.
	\newblock The communication complexity of pointer chasing.
	\newblock {\em J. Comput. Syst. Sci.}, 62(2):323--355, 2001.
	
	\bibitem[SWZ25]{SWZ25}
	William Swartworth, David~P. Woodruff, and Samson Zhou.
	\newblock Perfect lp sampling with polylogarithmic update time.
	\newblock In {\em {FOCS}}, pages 1936--1960. {IEEE}, 2025.
	
	\bibitem[WXZ25]{WXZ25}
	David~P. Woodruff, Shenghao Xie, and Samson Zhou.
	\newblock Perfect sampling in turnstile streams beyond small moments.
	\newblock {\em Proc. {ACM} Manag. Data}, 3(2):106:1--106:27, 2025.
	
	\bibitem[Zha25]{Zhang25}
	Qin Zhang.
	\newblock Robust statistical analysis on streaming data with near-duplicates in
	general metric spaces.
	\newblock {\em Proc. {ACM} Manag. Data}, 3(2):111:1--111:25, 2025.
	
	\bibitem[Zha26a]{Zhang26a}
	Qin Zhang.
	\newblock Frequency moments beyond equality: Streaming cosine density moments,
	2026.
	\newblock arXiv:2609.06925.
	
	\bibitem[Zha26b]{Zhang26b}
	Qin Zhang.
	\newblock Streaming algorithms for gaussian kernel density statistics, 2026.
	\newblock arXiv:2609.09622.
	
\end{thebibliography}

\newcommand{\etalchar}[1]{$^{#1}$}

\appendix

\section{Preliminaries}
\label{sec:preliminaries}

\vspace{2mm}
\noindent{\bf  Information theory basics.\ }
Let $H(X)$ denote the Shannon entropy of a random variable $X$. The conditional entropy of $X$ given $Y$ is denoted by $H(X \mid Y)$. The mutual information between $X$ and $Y$ is written as $I(X; Y)$, and the conditional mutual information given $Z$ is denoted by $I(X; Y \mid Z)$. We summarize below the properties of entropy and mutual information that will be used. Let $X, Y, Z, W$ be random variables.  
\begin{enumerate}
	\item If $X$ takes values in $\{1,2,\ldots,m\}$, then $H(X) \in [0, \log m]$.
	
	\item $H(X) \ge H(X \mid Y)$ and $I(X; Y) = H(X) - H(X \mid Y) \ge 0$.
	
	\item Chain rule of mutual information:
	$
	I(X, Y; Z) = I(X; Z) + I(Y; Z \mid X).
	$
	
	\item Subadditivity of conditional entropy:
	$
	H(X,Y \mid Z) \;\le\; H(X \mid Z) + H(Y \mid Z).
	$
	
	\item If $X$ and $Z$ are independent, then $I(X; Y \mid Z) \ge I(X; Y)$. Likewise, if $X$ and $Z$ are independent given $W$, then $I(X; Y \mid Z, W) \ge I(X; Y \mid W)$.
\end{enumerate}

\vspace{2mm}
\noindent{\bf Probabilistic tools.\ } 
\begin{lemma}[Chernoff--Hoeffding bounds]
	\label{lem:chernoff}
	Let $X = \sum_{i \in [m]} X_i$ be a sum of independent random variables
	taking values in $[0,1]$, and let $\mu = \E[X]$. Then for $\delta \in [0,1]$,
	$$
	\Pr[X \ge (1+\delta)\mu] \le e^{-\delta^2\mu/3}
	\quad \text{and} \quad
	\Pr[X \le (1-\delta)\mu] \le e^{-\delta^2\mu/2},
	$$
	and for $\delta \ge 1$,
	$$
	\Pr[X \ge (1+\delta)\mu] \le e^{-\delta\mu/3}.
	$$
	Moreover, for every $a > 0$,
	$$
	\Pr[\abs{X-\mu} \ge a] \le 2e^{-2a^2/m}.
	$$
	The upper-tail bounds remain valid when $\mu$ is replaced by any $\mu_0\ge\E[X]$, and the lower-tail bound when $\mu$ is replaced by any $\mu_0\le\E[X]$.
\end{lemma}

\begin{lemma}[Pinsker's inequality]
	\label{lem:pinsker}
	Let $P$ and $Q$ be probability distributions on a common finite set, and let
	$\TV(P,Q) = \frac{1}{2}\sum_x \abs{P(x)-Q(x)}$. If $D_{\mathrm{KL}}(P \Vert Q)$
	is measured in bits, then
	$$
	\TV(P,Q) \le \sqrt{\frac{\ln 2}{2} D_{\mathrm{KL}}(P \Vert Q)}.
	$$
\end{lemma}

\begin{lemma}[Fano's inequality]
	\label{lem:fano}
	Let $W$ be a bit and $V$ a random variable jointly distributed with $W$.
	Suppose there is a predictor $\hat{W} = f(V)$ with
	$\Pr[\hat{W} \ne W] \le \delta$ for some $\delta \le 1/2$. Then
	$$
	H(W \mid V) \le h_2(\delta),
	$$
	where $h_2(\delta) = -\delta\log_2\delta - (1-\delta)\log_2(1-\delta)$
	is the binary entropy function.
\end{lemma}

\begin{lemma}[Jensen's inequality]
	\label{lem:jensen}
	Let $\varphi$ be concave on an interval containing the range of a random variable $X$ with $\E[\abs{X}] < \infty$. Then
	$$
	\E[\varphi(X)] \le \varphi(\E[X]),
	$$
	with the inequality reversed if $\varphi$ is convex.
\end{lemma}


\end{document}